\documentclass[11pt]{article}
\usepackage[margin=1in]{geometry}

\usepackage{amsmath, amsthm, amssymb, microtype, mathrsfs, enumerate, stackrel, tikz, framed}
\usepackage[colorlinks, citecolor=blue, linkcolor=blue, bookmarks=true]{hyperref}
\usepackage[nameinlink, capitalize]{cleveref}

\newtheorem{theorem}{Theorem}
\newtheorem{lemma}{Lemma}
\newtheorem{corollary}{Corollary}
\newtheorem{proposition}{Proposition}
\newtheorem{claim}{Claim}
\theoremstyle{definition}
\newtheorem{definition}{Definition}

\crefname{claim}{Claim}{Claims}
\Crefname{claim}{Claim}{Claims}

\newcommand{\N}{\mathbb{N}}

\newcommand{\F}{\mathbb{F}}

\renewcommand{\emptyset}{\varnothing}
\renewcommand{\epsilon}{\varepsilon}

\renewcommand{\tilde}{\widetilde}
\renewcommand{\bar}{\overline}

\DeclareMathOperator{\poly}{poly}
\DeclareMathOperator*{\E}{E}
\DeclareMathOperator{\polylog}{polylog}

\DeclareMathOperator{\size}{size}
\DeclareMathOperator{\length}{length}
\DeclareMathOperator{\density}{density}
\DeclareMathOperator{\gip}{GIP}
\DeclareMathOperator{\queries}{queries}

\DeclareMathOperator{\stconn}{STConn}

\newcommand{\botn}{\bot_{\normalfont\text{n}}}
\newcommand{\botr}{\bot_{\normalfont\text{r}}}
\newcommand{\boti}{\bot_{\normalfont\text{i}}}

\makeatletter
\crefdefaultlabelformat{\leavevmode\sw@slant#2{\upshape#1}#3}
\makeatother

\title{Typically-Correct Derandomization for Small Time and Space}
\author{William M. Hoza\thanks{Supported by the NSF GRFP under Grant DGE-1610403 and by a Harrington Fellowship from UT Austin.}\\Department of Computer Science,\\University of Texas at Austin\\\texttt{whoza@utexas.edu}}

\begin{document}
	\maketitle
	
	\begin{abstract}
		Suppose a language $L$ can be decided by a bounded-error randomized algorithm that runs in space $S$ and time $n \cdot \poly(S)$. We give a randomized algorithm for $L$ that still runs in space $O(S)$ and time $n \cdot \poly(S)$ that uses only $O(S)$ random bits; our algorithm has a low failure probability on all but a negligible fraction of inputs of each length. As an immediate corollary, there is a deterministic algorithm for $L$ that runs in space $O(S)$ and succeeds on all but a negligible fraction of inputs of each length. We also give several other complexity-theoretic applications of our technique.
	\end{abstract}
	
	\section{Introduction}
	
	\subsection{The power of randomness when time and space are limited} \label{sec:bptisp-intro}
	
	A central goal of complexity theory is to understand the relationship between three fundamental resources: time, space, and randomness. Based on a long line of research \cite{yao82, bm84, bfnw93, nw94, iw97, stv01, kvm02}, most complexity theorists believe that randomized decision algorithms can be made deterministic without paying too much in terms of time and space. Specifically, suppose a language $L$ can be decided by a randomized algorithm that runs in time $T = T(n) \geq n$ and space $S = S(n) \geq \log n$. Klivans and van Melkebeek showed that assuming some language in $\mathbf{DSPACE}(n)$ has exponential circuit complexity, there is a deterministic algorithm for $L$ that runs in time $\poly(T)$ and space $O(S)$ \cite{kvm02}.\footnote{More generally, Klivans and van Melkebeek constructed a pseudorandom generator that fools size-$T$ circuits on $T$ input bits under this assumption. The generator has seed length $O(\log T)$ and is computable in $O(\log T)$ space.}
	
	Proving the hypothesized circuit lower bound seems unlikely for the foreseeable future. In the 90s and early 2000s, researchers managed to prove powerful \emph{unconditional} derandomization theorems by focusing on the space complexity of the deterministic algorithm. For example, Nisan and Zuckerman showed that if $S \geq T^{\Omega(1)}$, there is a deterministic algorithm for $L$ that runs in space $O(S)$ \cite{nz96}.\footnote{More generally, the Nisan-Zuckerman theorem applies as long as the original randomized algorithm for $L$ uses only $\poly(S)$ random bits, regardless of how much time it takes.} Alas, in the past couple of decades, progress on such general, unconditional derandomization has stalled. Nobody has managed to extend the Nisan-Zuckerman theorem to a larger regime of pairs $(T, S)$, and researchers have been forced to focus on more restricted models of computation.
	
	In this paper, we focus on \emph{highly efficient} randomized algorithms. That is, we consider the case that $T$ and $S$ are both small, such as $T \leq \tilde{O}(n)$ and $S \leq O(\log n)$.
	
	\subsection{Our results}
	
	\subsubsection{Reducing the amount of randomness to $O(S)$}
	Suppose $T \leq n \cdot \poly(S)$. For our main result, we give a randomized algorithm for $L$ that still runs in time $n \cdot \poly(S)$ and space $O(S)$ that uses only $O(S)$ random bits. The catch is that our algorithm is only guaranteed to succeed on \emph{most} inputs. The fraction of ``bad'' inputs of length $n$ is at most $2^{-S^c}$, where $c \in \N$ is an arbitrarily large constant. On ``good'' inputs, our algorithm's failure probability is at most $2^{-S^{1 - \alpha}}$, where $\alpha > 0$ is an arbitrarily small constant.
	
	\subsubsection{Eliminating randomness entirely}
	From the result described in the preceding paragraph, a deterministic algorithm that runs in space $O(S)$ follows immediately by iterating over all $O(S)$-bit random strings. We can express this theorem in terms of complexity classes using terminology introduced by Kinne et al.\ for typically-correct algorithms \cite{kvms12}. Suppose $L$ is a language, $\mathbf{C}$ is a complexity class, and $\epsilon(n)$ is a function. We say that \emph{$L$ is within $\epsilon$ of $\mathbf{C}$} if there is some $L' \in \mathbf{C}$ such that for every $n$,
	\begin{equation}
	\Pr_{x \in \{0, 1\}^n}[x \in L \Delta L'] \leq \epsilon(n).
	\end{equation}
	If $\mathbf{C}$ and $\mathbf{C}'$ are complexity classes, we say that \emph{$\mathbf{C}$ is within $\epsilon$ of $\mathbf{C}'$} if every language in $\mathbf{C}$ is within $\epsilon$ of $\mathbf{C}'$. In these terms, our result is that
	\begin{equation} \label{eqn:bptisp}
		\mathbf{BPTISP}(n \cdot \poly(S), S) \text{ is within } 2^{-S^c} \text{ of } \mathbf{DSPACE}(S).
	\end{equation}
	Here, $\mathbf{BPTISP}(T, S)$ is the class of languages that can be decided by a bounded-error randomized algorithm that runs in time $O(T(n))$ and space $O(S(n))$, and $\mathbf{DSPACE}(S)$ is the class of languages that can be decided by a deterministic algorithm that runs in space $O(S)$. Note that if $S \geq n^{\Omega(1)}$, the mistake rate in \cref{eqn:bptisp} drops below $2^{-n}$. Since there are only $2^n$ inputs of length $n$, the algorithm must in fact be correct on all inputs. Our result can therefore be viewed as a generalization of the Nisan-Zuckerman theorem $\mathbf{BPTISP}(\poly(S), S) \subseteq \mathbf{DSPACE}(S)$ \cite{nz96}.
	
%
	
	\subsubsection{Derandomization with advice}
	
	Adleman's argument \cite{adl78} shows that $\mathbf{BPL} \subseteq \mathbf{L}/\poly$. We study the problem of derandomizing $\mathbf{BPL}$ with as little advice as possible. Goldreich and Wigderson discovered a critical threshold: roughly, if an algorithm can be derandomized with fewer than $n$ bits of advice, then there is a \emph{typically-correct} derandomization of the algorithm with \emph{no} advice \cite{gw02}.\footnote{This result also requires that (a) most advice strings are ``good'', and (b) there is an appropriate efficient extractor.}
	
	Motivated by this phenomenon, Fortnow and Klivans proved that $\mathbf{BPL} \subseteq \mathbf{L}/O(n)$ \cite{fk06}. We refine their argument and show that $\mathbf{BPL} \subseteq \mathbf{L}/(n + O(\log^2 n))$, getting very near the critical threshold of $n$ bits of advice. More interestingly, we show that the connection identified by Goldreich and Wigderson \cite{gw02} works the other way: in the space-bounded setting, typically-correct derandomizations imply derandomizations with just a little advice. Combining with our main result gives that for every constant $c \in \N$,
	\begin{equation}
		\mathbf{BPTISP}(\tilde{O}(n), \log n) \subseteq \mathbf{L}/(n - \log^c n).
	\end{equation}
	
	\subsubsection{Derandomizing Turing machines} \label{sec:tm-overview}
	All algorithms in the results mentioned so far are formulated in a general \emph{random-access} model, i.e., the algorithm can read any specified bit of its input in a single step. (See \cref{sec:model} for details.) We also study the weaker \emph{multitape Turing machine} model. The main weakness of the Turing machine model is that if its read head is at position $i$ of its input and it wishes to read bit $j$ of its input, it must spend $|i - j|$ steps moving its read head to the appropriate location. Let $\mathbf{BPTISP}_{\text{TM}}(T, S)$ denote the class of languages that can be decided by a bounded-error randomized Turing machine that runs in time $O(T(n))$ and space $O(S(n))$.
	
	\paragraph{Beyond linear advice}
	We give a typically-correct derandomization for $\mathbf{BPTISP}_{\text{TM}}$ analogous to our main result but with a lower mistake rate. In terms of advice, our derandomization implies that for every constant $c \in \N$,
	\begin{equation} \label{eqn:bptisp-tm-advice}
		\mathbf{BPTISP}_{\text{TM}}(\tilde{O}(n), \log n) \subseteq \mathbf{L}/O\left(\frac{n}{\log^c n}\right).
	\end{equation}
	\Cref{eqn:bptisp-tm-advice} gives an interesting example of a class of $\mathbf{BPL}$ algorithms that can be derandomized with $o(n)$ bits of advice.
	
	\paragraph{Beyond quasilinear time}
	Using different techniques, we also show how to derandomize log-space Turing machines that use almost a \emph{quadratic} amount of time. In particular, we show that if $TS^2 \leq o(n^2 / \log n)$, then
	\begin{equation} \label{eqn:bptisp-tm-cc}
	\mathbf{BPTISP}_{\text{TM}}(T, S) \text{ is within } o(1) \text{ of } \mathbf{DTISP}(\poly(n), S).
	\end{equation}
	
	\subsubsection{Disambiguating nondeterministic algorithms}
	
	For some of our derandomization results, we give analogous theorems regarding \emph{unambiguous} simulations of \emph{nondeterministic} algorithms. We defer a discussion of these results to \cref{sec:nl-ul}.
	
	\subsection{Techniques}
	
	\subsubsection{``Out of sight, out of mind''} \label{sec:out-of-sight}
	
	Our typically-correct derandomizations work by treating the \emph{input as a source of randomness}. This idea was pioneered by Goldreich and Wigderson \cite{gw02}. For the sake of discussion, let $\mathcal{A}$ be a randomized algorithm that uses $n$ random bits. A na{\"i}ve strategy for derandomizing $\mathcal{A}$ is to run $\mathcal{A}(x, x)$. Most random strings of $\mathcal{A}$ lead to the right answer, so it is tempting to think that for most $x$, $\mathcal{A}(x, x)$ will give the right answer. This reasoning is flawed, because $\mathcal{A}$ might behave poorly when its input is \emph{correlated} with its random bits.
	
	In this work, we avoid these troublesome correlations using a simple idea embodied by the adage ``out of sight, out of mind.'' We use \emph{part} of the input as a source of randomness while $\mathcal{A}$ is processing \emph{the rest} of the input.
	
	To go into more detail, suppose $\mathcal{A}$ runs in time $\tilde{O}(n)$ and space $O(\log n)$. Our randomness-efficient simulation of $\mathcal{A}$ operates in $\polylog(n)$ \emph{phases}. At the beginning of a new phase, we pick a random $\polylog(n)$-bit block $x \vert_I$ of the input $x$. We apply a seeded extractor to $x \vert_I$, giving a string of length $\Theta(\log^2 n)$. We apply Nisan's pseudorandom generator for space-bounded computation \cite{nis92}, giving a pseudorandom string of length $\tilde{O}(n)$. We use the pseudorandom string to run the simulation of $\mathcal{A}$ forward until it tries to read from $x \vert_I$, at which time we pause the simulation of $\mathcal{A}$ and move on to the next phase.
	
	The key point is that the output of the extractor is processed without ever looking at $x \vert_I$, the input to the extractor. Extractors are good \emph{samplers} \cite{zuc97}, and $\mathcal{A}$ only has polynomially many possible configurations, so for most $x$, the output of the extractor is essentially as good as a uniform random seed to Nisan's generator. Therefore, in each phase, with high probability, we successfully simulate $n/\polylog(n)$ steps of $\mathcal{A}$ before it reads from $x \vert_I$ and we have to move on to the next phase. Thus, with high probability, after $\polylog(n)$ phases, the simulation of $\mathcal{A}$ is complete.
	
	Each bit of the output of Nisan's generator can be computed in time\footnote{See work by Diehl and van Melkebeek \cite{dvm06} for an even faster implementation of Nisan's generator.} $\polylog(n)$ and space $O(\log n)$. Therefore, our simulation of $\mathcal{A}$ still runs in time $\tilde{O}(n)$ and space $O(\log n)$, but now it uses just $\polylog(n)$ random bits ($O(\log n)$ random bits per phase to pick the random block $I$ and to pick a seed for the extractor).
	
	The reader may wonder whether we could have achieved the same effect by simply directly applying Nisan's generator from the start -- its seed length is $\polylog(n)$, after all. The point is that Nisan's generator requires \emph{two-way} access to its seed, whereas our simulation only uses one-way access to its random bits. During our simulation, we are able to give Nisan's generator two-way access to its seed, because we have two-way access to the \emph{input} $x$ from which we extract that seed.
	
	Finally, because our simulation reads its $\polylog(n)$ random bits from left to right, we can further reduce the number of random bits to just $O(\log n)$ by applying the Nisan-Zuckerman pseudorandom generator \cite{nz96}.
	
	\subsubsection{Other techniques}
	
	Our derandomizations with advice are based on Fortnow and Klivans' technique for proving $\mathbf{BPL} \subseteq \mathbf{L}/O(n)$ \cite{fk06} and Nisan's technique for proving $\mathbf{RL} \subseteq \mathbf{SC}$ \cite{nis94}. Our derandomization of $\mathbf{BPTISP}_{\text{TM}}$ with a low mistake rate uses a similar ``out of sight, out of mind'' technique as our main result. The lower mistake rate is achieved by exploiting knowledge of the region of the input that will be processed in the near future, based on the locality of the Turing machine's read head. Our derandomization of $\mathbf{BPTISP}_{\text{TM}}(T, S)$ for $T(n) \approx n^2$ is based on a seed-extending pseudorandom generator for multiparty communication protocols by Kinne et al.\ \cite{kvms12}.
	
	\subsection{Related work}
	
	We will only mention some highlights of the large body of research on unconditional derandomization of time- and space-bounded computation. Fix $L \in \mathbf{BPTISP}(T, S)$. Nisan gave a randomized algorithm for $L$ that runs in time $\poly(T)$ and space $O(S \log T)$ that uses only $O(S \log T)$ random bits \cite{nis92}. Nisan also gave a deterministic algorithm for $L$ that runs in time $2^{O(S)}$ and space $O(S \log T)$ \cite{nis94}. Nisan and Zuckerman gave a randomized algorithm for $L$ that runs in time $\poly(T)$ and space $O(S + T^{\epsilon})$ that uses only $O(S + T^{\epsilon})$ random bits, where $\epsilon > 0$ is an arbitrarily small constant \cite{nz96} (this is a generalization of the result mentioned in \cref{sec:bptisp-intro}). Saks and Zhou gave a deterministic algorithm for $L$ that runs in space $O(S \sqrt{\log T})$ \cite{sz99}. Combining the techniques from several of these works, Armoni \cite{arm98} gave a deterministic algorithm for $L$ that runs in space\footnote{Actually, the space bound given in \cref{eqn:arm98} is achieved by using better extractors than were known when Armoni wrote his paper \cite{arm98, knw08}.}
	\begin{equation} \label{eqn:arm98}
		O\left(S \cdot \sqrt{\frac{\log T}{\max\{1, \log S - \log \log T\}}}\right).
	\end{equation}
	Armoni's algorithm remains the most space-efficient derandomization known for all $T$ and $S$. When $T = \tilde{\Theta}(n)$ and $S = \Theta(\log n)$, Armoni's algorithm runs in space $\Theta(\log^{3/2} n)$, just like the earlier Saks-Zhou algorithm \cite{sz99}. Cai et al.\ gave a time-space tradeoff \cite{ccvm06} interpolating between Nisan's deterministic algorithm \cite{nis94} and the Saks-Zhou algorithm \cite{sz99}.
	
	All of the preceding results apply, \emph{mutatis mutandis}, to derandomizing algorithms that use at most $T$ random bits, regardless of how much time they take. In contrast, our proofs crucially rely on the fact that a time-$T$ algorithm queries its \emph{input} at most $T$ times. This aspect of our work is shared by work by Beame et al.\ \cite{bssv03} on time-space lower bounds. 
	
	Goldreich and Wigderson's idea of using the input as a source of randomness for a typically-correct derandomization \cite{gw02} has been applied and developed by several researchers \cite{at04, vms05, ks05, zim08, sha11, kvms12, sw14, alm19}; see related survey articles by Shaltiel \cite{sha10} and by Hemaspaandra and Williams \cite{hw12}. Researchers have proven unconditional typically-correct derandomization results for several restricted models, including sublinear-time algorithms \cite{zim08, sha11}, communication protocols \cite{sha11, kvms12}, constant-depth circuits \cite{sha11, kvms12}, and streaming algorithms \cite{sha11}. On the other hand, Kinne et al.\ proved that any typically-correct derandomization of $\mathbf{BPP}$ with a sufficiently low mistake rate would imply strong circuit lower bounds \cite{kvms12}. We are the first to study typically-correct derandomization for algorithms with simultaneous bounds on time and space.
	
	\subsection{Outline of this paper}
	In \cref{sec:prelim}, we discuss random-access models of computation and extractors. In \cref{sec:bptisp}, we give our derandomization of $\mathbf{BPTISP}(n \cdot \poly(S), S)$. In \cref{sec:bptisp-tm}, we give our two derandomizations of $\mathbf{BPTISP}_{\text{TM}}(T, S)$. 
	In \cref{sec:bpl-advice}, we discuss derandomization with advice. \Cref{sec:nl-ul} concerns disambiguation of nondeterministic algorithms, and we conclude in \cref{sec:future} with some suggested directions for further research.
	
	\section{Preliminaries} \label{sec:prelim}
	
	\subsection{General notation}
	
	\paragraph{Strings}
	For strings $x, y$, let $x \circ y$ denote the concatenation of $x$ with $y$. For a natural number $n$, let $[n] = \{1, 2, \dots, n\}$. For a string $x \in \{0, 1\}^n$ and a set $I = \{i_1 < i_2 < \dots < i_{\ell}\} \subseteq [n]$, let $x \vert_I = x_{i_1} x_{i_2} \dots x_{i_{\ell}} \in \{0, 1\}^{\ell}$.
	
	\paragraph{Sets}
	For a finite set $X$, we will use the notations $\#X$ and $|X|$ interchangeably to refer to the number of elements of $X$. For $X \subseteq \{0, 1\}^n$, let $\density(X) = |X|/2^n$. We will sometimes omit the parentheses, e.g., $\density\{000, 111\} = 0.25$. We identify a language $L \subseteq \{0, 1\}^*$ with its indicator function $L: \{0, 1\}^* \to \{0, 1\}$, i.e.,
	\begin{equation}
	L(x) = \begin{cases}
	1 & \text{if } x \in L \\
	0 & \text{if } x \not \in L.
	\end{cases}
	\end{equation}
	
	\paragraph{Probability}
	If $X$ and $Y$ are probability distributions on the same space, we write $X \sim_{\epsilon} Y$ to indicate that $X$ and $Y$ are $\epsilon$-close in total variation distance. For $T \in \N$, let $U_T$ denote the uniform distribution over $\{0, 1\}^T$.
	
	\subsection{Random-access algorithms} \label{sec:model}
	
	Our main theorems govern general \emph{random-access algorithms}. Our results are not sensitive to the specific choice of model of random-access computation. For concreteness, following Fortnow and van Melkebeek \cite{fvm00}, we will work with the \emph{random-access Turing machine} model. This model is defined like the standard multitape Turing machine model, except that each ordinary tape is supplemented with an ``index tape'' that can be used to move the ordinary tape's head to an arbitrary specified location in a single step. See the paper by Fortnow and van Melkebeek \cite{fvm00} for details.
	
	A \emph{randomized random-access Turing machine} is a random-access Turing machine equipped with an additional read-only tape, initialized with random bits, that can only be read from \emph{left to right}. Thus, if the algorithm wishes to reread old random bits, it needs to have copied them to a work tape, which counts toward the algorithm's space usage. The random tape does not have a corresponding index tape. 
	
	For functions $T: \N \to \N$ and $S: \N \to \N$, we define $\mathbf{BPTISP}(T, S)$ to be the class of languages $L$ such that there is a randomized random-access Turing machine $\mathcal{A}$ such that on input $x \in \{0, 1\}^n$, $\mathcal{A}(x)$ always halts in time $O(T(n))$, $\mathcal{A}(x)$ always touches $O(S(n))$ cells on all of its read-write tapes, and $\Pr[\mathcal{A}(x) = L(x)] \geq 2/3$.

	\subsection{Randomized branching programs}
	Our algorithms are most naturally formulated in terms of branching programs, a standard \emph{nonuniform} model of time- and space-bounded computation. Recall that in a digraph, a \emph{terminal vertex} is a vertex with no outgoing edges. In the following definition, $n$ is the number of input bits and $m$ is the number of random bits.

	\begin{definition}
		A \emph{randomized branching program} on $\{0, 1\}^n \times \{0, 1\}^m$ is a directed acyclic graph, where each nonterminal vertex $v$ is labeled with two indices $i(v) \in [n], j(v) \in [m]$ and has four outgoing edges labeled with the four two-bit strings. If $\mathcal{P}$ is a randomized branching program, we let $V(\mathcal{P})$ be the set of vertices of $\mathcal{P}$.
	\end{definition}

	The interpretation is that from vertex $v$, the program follows the edge labeled $x_{i(v)} y_{j(v)}$, where $x$ is the input and $y$ is the random string. This interpretation is formalized by the following definition, which sets $\mathcal{P}(v; x, y)$ to be the vertex reached from $v$ on input $x$ using randomness $y$.
	
	\begin{definition} \label{def:randomized-branching-program-operation}
		Suppose $\mathcal{P}$ is a randomized branching program on $\{0, 1\}^n \times \{0, 1\}^m$. We identify $\mathcal{P}$ with a function $\mathcal{P}: V(\mathcal{P}) \times \{0, 1\}^n \times \{0, 1\}^m \to V(\mathcal{P})$ defined as follows. Fix $v \in V(\mathcal{P}), x \in \{0, 1\}^n, y \in \{0, 1\}^m$. Take a walk through $\mathcal{P}$ by starting at $v$ and, having reached vertex $u$, following the edge labeled $x_{i(u)} y_{j(u)}$. Then $\mathcal{P}(v; x, y)$ is the terminal vertex reached by this walk.
	\end{definition}

	As previously discussed, random-access Turing machines can only access their random bits from left to right. This corresponds to an \emph{R-OW randomized branching program}.
	\begin{definition}
		An \emph{R-OW randomized branching program} is a randomized branching program $\mathcal{P}$ such that for every edge $(v, v')$ between two nonterminal vertices, $j(v') \in \{j(v), j(v) + 1\}$.
	\end{definition}
	The term ``R-OW'' indicates that the branching program has ``random access'' to its input bits and ``one-way access'' to its random bits.

	The \emph{size} of a branching program is defined as $\size(\mathcal{P}) = |V(\mathcal{P})|$. The \emph{length} of the program, $\length(\mathcal{P})$, is defined to be the length of the longest path through the program. Observe that $\mathbf{BPTISP}(T, S)$ corresponds to R-OW randomized branching programs of size $2^{O(S)}$ and length $O(T)$.
	
	Many of our algorithms will use a \emph{restriction} operation that we now introduce.
	\begin{definition}
		Suppose $\mathcal{P}$ is a randomized branching program on $\{0, 1\}^n$ and $I \subseteq [n]$. Let $\mathcal{P} \vert_I$ be the program obtained from $\mathcal{P}$ by deleting all outgoing edges from vertices $v$ such that $i(v) \not \in I$.
	\end{definition}
	
	So in $\mathcal{P} \vert_I$, there are two types of terminal vertices: vertices that were terminal in $\mathcal{P}$, and vertices $v$ that are now terminal because $i(v) \not \in I$. The computation $\mathcal{P} \vert_I(v; x, y)$ halts when it reaches either type of terminal vertex. Thus, $\mathcal{P}\vert_I(v; x, y)$ does not depend on $x \vert_{[n] \setminus I}$, because $\mathcal{P}\vert_I(v; x, y)$ outputs the vertex reached by running the computation $\mathcal{P}(v; x, y)$ until it finishes or it tries to read from $x \vert_{[n] \setminus I}$.

	\subsection{Extractors}
	
	Recall that a \emph{$(k, \epsilon)$-extractor} is a function $\mathsf{Ext}: \{0, 1\}^{\ell} \times \{0, 1\}^d \to \{0, 1\}^s$ such that if $X$ has ``min-entropy'' at least $k$ and $Y \sim U_d$ is independent of $X$, then $\mathsf{Ext}(X, Y) \sim_{\epsilon} U_s$. It can be shown nonconstructively that for every $\ell, k, \epsilon$, there exists $\mathsf{Ext}$ with $d \leq \log(\ell - k) + 2\log(1/\epsilon) + O(1)$ and $s \geq k + d - 2\log(1/\epsilon) - O(1)$ (see, e.g., Vadhan's monograph \cite{vad12}).
	
	We will need a computationally efficient extractor. The extractor literature has mainly focused on the time complexity of computing extractors, but we are concerned with space complexity, too. This paper is not meant to be about extractor constructions, so we encourage the reader to simply pretend that optimal extractors can be computed in a single step with no space overhead. In actuality, we will use two incomparable non-optimal extractors.
	
	To prove our main results, we will use an extractor by Shaltiel and Umans \cite{su05}. The benefit of the Shaltiel-Umans extractor is that it allows for small error $\epsilon$.
	\begin{theorem}[\cite{su05}] \label{thm:su05}
		Fix a constant $\alpha > 0$. For every $\ell, k \in \N, \epsilon > 0$ such that $k \geq \log^{4/\alpha} \ell$ and $k \geq \log^{4/\alpha}(1/\epsilon)$, there is a $(k, \epsilon)$-extractor $\mathsf{SUExt}: \{0, 1\}^{\ell} \times \{0, 1\}^d \to \{0, 1\}^s$ where $d \leq O\left(\log \ell + \frac{\log \ell \log(1/\epsilon)}{\log k}\right)$ and $s \geq k^{1 - \alpha}$. Given $x, y, k$, and $\epsilon$, $\mathsf{SUExt}(x, y)$ can be computed in time $\poly(\ell)$ and space $O(d)$.
	\end{theorem}
	
	To derandomize $\mathbf{BPL}$ with as little advice as possible, we will use an extractor by Guruswami, Umans, and Vadhan \cite{guv09} (not the most famous extractor from their work, but a slight variant). The benefit of the GUV extractor is that it outputs a constant fraction of the entropy.
	\begin{theorem}[\cite{guv09}] \label{thm:guv09}
		Let $\alpha, \epsilon > 0$ be constant. For every $\ell, k \in \N$, there is a $(k, \epsilon)$-extractor $\mathsf{GUVExt}: \{0, 1\}^{\ell} \times \{0, 1\}^d \to \{0, 1\}^s$ with $s \geq (1 - \alpha) k$ and $d \leq O(\log \ell)$ such that given $x$ and $y$, $\mathsf{GUVExt}(x, y)$ can be computed in $O(\log \ell)$ space.
	\end{theorem}

	In both cases, the original authors \cite{su05, guv09} did not explicitly analyze the space complexity of their extractors, so we explain in \cref{apx:su05,apx:guv09} why these extractors can be implemented in small space. (We remark that Hartman and Raz also constructed small-space extractors \cite{hr03}, but the seed lengths of their extractors are too large for us.)
%
%

	\subsubsection{Extractors as samplers} We will actually only be using extractors for their \emph{sampling} properties. The connection between extractors and samplers was first discovered by Zuckerman \cite{zuc97}. The following standard proposition expresses this connection for non-Boolean functions.

	\begin{proposition}[\cite{zuc97}] \label{prop:zuc97}
		Suppose $\mathsf{Ext}: \{0, 1\}^{\ell} \times \{0, 1\}^d \to \{0, 1\}^s$ is a $(k, \epsilon)$-extractor and $f: \{0, 1\}^s \to V$ is a function. Let $\delta = \epsilon |V| / 2$. Then
		\begin{equation} \label{eqn:zuc97}
			\#\{x \in \{0, 1\}^{\ell} : f(U_s) \not \sim_\delta f(\mathsf{Ext}(x, U_d))\} \leq 2^{k + 1} |V|.
		\end{equation}
	\end{proposition}

	For completeness, we include a proof of \cref{prop:zuc97} in \cref{apx:zuc97}, since the specific statement of \cref{prop:zuc97} does not appear in Zuckerman's paper \cite{zuc97}.
	
	\subsection{Constructibility}
	We say that $f: \N \to \N$ is \emph{constructible} in space $S(n)$, time $T(n)$, etc. if there is a deterministic random-access Turing machine $\mathcal{A}$ that runs in the specified resource bounds with $\mathcal{A}(1^n) = f(n)$, written in binary. As usual, we say that $f$ is \emph{space constructible} if $f$ is constructible in space $O(f(n))$. We say that $\delta: \N \to [0, 1]$ is constructible in specified resource bounds if $\delta$ can be written as $\delta(n) = \frac{\delta_1(n)}{\delta_2(n)}$, where $\delta_1, \delta_2: \N \to \N$ are both constructible in the specified resource bounds.
	
	\section{Derandomizing efficient random-access algorithms} \label{sec:bptisp} 
	
	\subsection{Main technical algorithm: Low-randomness simulation of branching programs}
	Suppose $\mathcal{P}$ is an R-OW randomized branching program on $\{0, 1\}^n \times \{0, 1\}^T$ of length $T$ and size $2^S$. (As a reminder, such a program models $\mathbf{BPTISP}(T, S)$.) Given $\mathcal{P}$, $v_0$, and $x$, the distribution $\mathcal{P}(v_0; x, U_T)$ can trivially be sampled in time $T \cdot \poly(S)$ and space $O(S)$ using $T$ random bits. Our main technical result is an efficient typically-correct algorithm for approximately sampling $\mathcal{P}(v_0; x, U_T)$ using roughly $T/n$ random bits.
	
	\begin{samepage}
	\begin{theorem} \label{thm:simulate-branching-program}
		For each constant $c \in \N$, there is a randomized algorithm $\mathsf{A}$ with the following properties. Suppose $\mathcal{P}$ is an R-OW randomized branching program on $\{0, 1\}^n \times \{0, 1\}^T$ with $S \geq \log n$, where $S \stackrel{\normalfont\text{def}}{=} \lceil \log \size(\mathcal{P}) \rceil$. Suppose $v_0 \in V(\mathcal{P})$, $T \geq \length(\mathcal{P})$, and $x \in \{0, 1\}^n$. Then $\mathsf{A}(\mathcal{P}, v_0, x, T)$ outputs a vertex $v \in V(\mathcal{P})$ in time\footnote{The graph of $\mathcal{P}$ should be encoded in adjacency list format, so that the neighborhood of a vertex $v$ can be computed in $\poly(S)$ time.} $T \cdot \poly(S)$ and space $O(S)$ using $\lceil T/n \rceil \cdot \poly(S)$ random bits. Finally, for every such $\mathcal{P}, v_0, T$,
		\begin{equation} \label{eqn:simulate-branching-program-correctness}
			\density\{x \in \{0, 1\}^n : \mathsf{A}(\mathcal{P}, v_0, x, T) \not \sim_{\exp(-cS)} \mathcal{P}(v_0; x, U_T)\} \leq 2^{-S^c}.
		\end{equation}
	\end{theorem}
	\end{samepage}

	The algorithm of \cref{thm:simulate-branching-program} relies on Nisan's pseudorandom generator \cite{nis92}. The seed length of Nisan's generator is not $O(S)$, but Nisan's generator does \emph{run} in space $O(S)$, given two-way access to the seed.
	\begin{theorem}[\cite{nis92}] \label{thm:nis92}
		For every $S, T \in \N, \epsilon > 0$ with $T \leq 2^S$, there is a generator $\mathsf{NisGen}: \{0, 1\}^s \to \{0, 1\}^T$ with seed length $s \leq O((S + \log(1/\epsilon)) \cdot \log T)$, such that if $\mathcal{P}$ is an R-OW randomized branching program of size $2^S$, $v$ is a vertex, and $x$ is an input, then
		\begin{equation}
			\mathcal{P}(v; x, \mathsf{NisGen}(U_s)) \sim_{\epsilon} \mathcal{P}(v; x, U_T).
		\end{equation}
		Given $S, T, \epsilon, z, i$, the bit $\mathsf{NisGen}(z)_i$ can be computed in time $\poly(S, \log(1/\epsilon))$ and space $O(S + \log(1/\epsilon))$.
	\end{theorem}

	\begin{figure}
		\begin{framed}
			\begin{enumerate}
				\item If $S^{c + 1} > \lfloor n/9 \rfloor$, directly simulate $P(v_0; x, U_T)$ using $T$ random bits. Otherwise:
				\item Let $I_1, I_2, \dots, I_B \subseteq [n]$ be disjoint sets of size $S^{c + 1}$ with $B$ as large as possible.
				\item Initialize $v = v_0$. Repeat $r$ times, where $r$ is given by \cref{eqn:r-def}:
				\begin{enumerate}
					\item Pick $b \in [B]$ uniformly at random and let $I = I_b$.
					\item Pick $y \in \{0, 1\}^{O(S)}$ uniformly at random.
					\item Let $v = \mathcal{P}\vert_{[n] \setminus I}(v; x, \mathsf{NisGen}(\mathsf{SUExt}(x\vert_I, y)))$.
				\end{enumerate}
				\item Output $v$.
			\end{enumerate}
			\vspace{-5mm}
		\end{framed}
		\vspace{-5mm}
		\caption{The algorithm $\mathsf{A}$ of \cref{thm:simulate-branching-program}.} \label{fig:simulate-branching-program}
	\end{figure}
	For \cref{thm:simulate-branching-program}, we can replace $T$ with $\min\{T, 2^S\}$ without loss of generality, so we will assume that $T \leq 2^S$. The algorithm $\mathsf{A}$ is given in \cref{fig:simulate-branching-program}.
	
	\paragraph{Parameters} \label{par:parameters}
	Set
	\begin{equation} \label{eqn:r-def}
		r \stackrel{\text{def}}{=} \max\left\{\left\lceil \frac{8T}{B - 8} \right\rceil, 8(cS + 1)\right\} = \lceil T/n \rceil \cdot \poly(S).
	\end{equation}
	The parameter $r$ is the number of ``phases'' of $\mathsf{A}$ as outlined in \cref{sec:out-of-sight}. Note that if $S^{c + 1} \leq \lfloor n/9 \rfloor$, then $B \geq 9$, so \cref{eqn:r-def} makes sense. Naturally, Nisan's generator $\mathsf{NisGen}$ is instantiated with the parameters $S, T$ from the statement of \cref{thm:simulate-branching-program}. The error of $\mathsf{NisGen}$ is set at
	\begin{equation} \label{eqn:eps-def}
		\epsilon \stackrel{\text{def}}{=} \frac{e^{-cS}}{4r} = 2^{-\Theta(S)}.
	\end{equation}
	That way, the seed length of $\mathsf{NisGen}$ is $s \leq O(S \log T) \leq O(S^2)$. The algorithm $\mathsf{A}$ also relies on the Shaltiel-Umans extractor $\mathsf{SUExt}$ of \cref{thm:su05}. This extractor is instantiated with source length $\ell \stackrel{\text{def}}{=} S^{c + 1}$, $\alpha \stackrel{\text{def}}{=} 2/3$, error
	\begin{equation} \label{eqn:eps'-def}
		\epsilon' \stackrel{\text{def}}{=} \frac{e^{-cS}}{2r \cdot 2^S} = 2^{-\Theta(S)},
	\end{equation}
	and entropy
	\begin{equation}
		k \stackrel{\text{def}}{=} \max\{s^3, \log^6 \ell, \log^6(1/\epsilon)\} = \Theta(S^6).
	\end{equation}
	Our choice of $k$ explicitly meets the hypotheses of \cref{thm:su05}, and by construction, $k^{1 - \alpha} \geq s$, so we can think of $\mathsf{SUExt}$ as outputting $s$ bits.
	
	\paragraph{Efficiency}
	We now analyze the computational efficiency of $\mathsf{A}$. First, we bound the running time. If $S^{c + 1} > \lfloor n/9 \rfloor$, then $\mathsf{A}$ clearly runs in time $T \cdot \poly(S)$. Otherwise, $\mathsf{A}$ repeatedly replaces $v$ with one of its neighbors a total of at most $T$ times, since $T \geq \length(\mathcal{P})$. Each such step requires computing a bit of Nisan's generator, which takes time $\poly(S)$, times $\poly(S)$ steps to compute each bit of the seed of Nisan's generator by running $\mathsf{SUExt}$. Thus, overall, $\mathsf{A}$ runs in time $T \cdot \poly(S)$.
	
	Next, we bound the space complexity of $\mathsf{A}$. If $S^{c + 1} > \lfloor n/9 \rfloor$, then $\mathsf{A}$ clearly runs in space $O(S + \log T) = O(S)$. Otherwise, space is required to store a loop index ($O(\log r)$ bits), the vertex $v$ ($O(S)$ bits), the index $b$ ($O(\log n)$ bits), and the seed $y$ ($O(S)$ bits). These terms are all bounded by $O(S)$. Running $\mathsf{SUExt}$ takes $O(\log \ell + \frac{\log \ell \log(1/\epsilon')}{\log k})$ bits of space. Since $k \geq S^{\Omega(1)}$, $\frac{\log \ell}{\log k} \leq O(1)$, and hence the space used for $\mathsf{SUExt}$ is only $O(\log S + \log(1/\epsilon')) = O(S)$. Finally, running $\mathsf{NisGen}$ takes $O(S + \log(1/\epsilon)) = O(S)$ bits of space. Therefore, overall, $\mathsf{A}$ runs in space $O(S)$.
	
	Finally, we bound the number of random bits used by $\mathsf{A}$. If $S^{c + 1} > \lfloor n/9 \rfloor$, then $\mathsf{A}$ uses $T$ random bits, which is at most $\frac{9T(1 + S^{c + 1})}{n}$ in this case. Otherwise, in each iteration of the loop, $\mathsf{A}$ uses $O(\log n)$ random bits for $b$, plus $O(S)$ random bits for $y$. Therefore, overall, the number of random bits used by $\mathsf{A}$ is $O(rS)$, which is $\lceil T/n \rceil \cdot \poly(S)$.
	
	\paragraph{Correctness}
	We now turn to proving \cref{eqn:simulate-branching-program-correctness}. If $S^{c + 1} > \lfloor n/9 \rfloor$, then obviously $\mathsf{A}(\mathcal{P}, v_0, x, T) \sim \mathcal{P}(v_0; x, U_T)$. Assume, therefore, that $S^{c + 1} \leq \lfloor n/9 \rfloor$. The proof will be by a hybrid argument with three hybrid distributions. The first hybrid distribution is defined by the algorithm $\mathsf{H}_1$ given by \cref{fig:hybrid-1}.
	\begin{figure}
		\begin{framed}
			\begin{enumerate}
				\item Initialize $v = v_0$. Repeat $r$ times:
				\begin{enumerate}
					\item Pick $b \in [B]$ uniformly at random and let $I = I_b$.
					\item Pick $y' \in \{0, 1\}^s$ uniformly at random.
					\item Let $v = \mathcal{P}\vert_{[n] \setminus I}(v; x, \mathsf{NisGen}(y'))$.
				\end{enumerate}
				\item Output $v$.
			\end{enumerate}
			\vspace{-5mm}
		\end{framed}
		\vspace{-5mm}
		\caption{The algorithm $\mathsf{H}_1$ defining the first hybrid distribution used to prove \cref{eqn:simulate-branching-program-correctness}. The only difference between $\mathsf{A}$ and $\mathsf{H}_1$ is that $\mathsf{H}_1$ picks a uniform random seed for $\mathsf{NisGen}$, instead of extracting the seed from the input.} \label{fig:hybrid-1}
	\end{figure}
	
	We need a standard fact about Markov chains. Suppose $M$ and $M'$ are stochastic matrices (i.e., each row is a probability vector) of the same size. We write $M \sim_{\gamma} M'$ to mean that for each row index $i$, the probability distributions $M_i$ and $M'_i$ are $\gamma$-close in total variation distance.
	\begin{lemma} \label{lem:markov}
		If $M \sim_{\gamma} M'$, then $M^r \sim_{\gamma r} (M')^r$.
	\end{lemma}
	For a proof of \cref{lem:markov}, see, e.g., work by Saks and Zhou \cite[Proposition 2.3]{sz99}. We are now ready to prove that for most $x$, the behavior of $\mathsf{A}$ is statistically similar to the behavior of $\mathsf{H}_1$.
	
	\begin{claim}[$\mathsf{A} \approx \mathsf{H}_1$] \label{clm:a-hybrid-1}
		Let $\delta = \epsilon' \cdot r \cdot 2^{S - 1} = 2^{-\Theta(S)}$. Then
		\begin{equation}
		\density\{x \in \{0, 1\}^n : \mathsf{A}(\mathcal{P}, v_0, x, T) \not \sim_{\delta} \mathsf{H}_1(\mathcal{P}, v_0, x, T)\} \leq 2^{-S^c}.
		\end{equation}
	\end{claim}

	\begin{proof}
		Fix any $b \in [B]$ and $v \in V(\mathcal{P})$. Let $I = I_b$, and fix any $x' \in \{0, 1\}^n$ with $x' \vert_I = 0^{|I|}$. Define $f: \{0, 1\}^s \to V$ by
		\begin{equation}
			f(y') = \mathcal{P}\vert_{[n] \setminus I}(v; x', \mathsf{NisGen}(y')).
		\end{equation}
		By \cref{prop:zuc97},
		\begin{equation}
			\#\{x\vert_I \in \{0, 1\}^\ell : f(\mathsf{SUExt}(x \vert_I, U_d)) \not \sim_{\epsilon' 2^{S - 1}} f(U_s)\} \leq 2^{k + S + 1}.
		\end{equation}
		Therefore,
		\begin{equation}
			\#\{x \in \{0, 1\}^n : x \vert_{[n] \setminus I} = x' \vert_{[n] \setminus I} \text{ and } f(\mathsf{SUExt}(x \vert_I, U_d)) \not \sim_{\epsilon' 2^{S - 1}} f(U_s)\} \leq 2^{k + S + 1}.
		\end{equation}
		Now, let $M[x]$ be the $\size(\mathcal{P}) \times \size(\mathcal{P})$ stochastic matrix defined by
		\begin{equation}
			M[x]_{uv} = \Pr_{b, y}[\mathcal{P} \vert_{[n] \setminus I}(u; x, \mathsf{NisGen}(\mathsf{SUExt}(x\vert_I, y))) = v \text{ where } I = I_b].
		\end{equation}
		Let $M'[x]$ be the stochastic matrix defined by
		\begin{equation}
			M'[x]_{uv} = \Pr_{b, y'}[\mathcal{P}\vert_{[n] \setminus I}(u; x, \mathsf{NisGen}(y')) = v \text{ where } I = I_b].
		\end{equation}
		By summing over all $b, v, x'$, we find that
		\begin{align}
			\#\{x \in \{0, 1\}^n : M[x] \not \sim_{\epsilon' 2^{S - 1}} M'[x]\} &\leq B \cdot 2^S \cdot 2^{n - \ell} \cdot 2^{k + S + 1} \\
			&\leq 2^{n - S^{c + 1} + O(S^6)} \\
			&\leq 2^{n - S^c},
		\end{align}
		assuming $c \geq 6$ and $n$ is sufficiently large. If $M[x] \sim_{\epsilon' 2^{S - 1}} M'[x]$, then by \cref{lem:markov}, $M[x]^r \sim_{\delta} M'[x]^r$. The output of $\mathsf{A}$ is a sample from $(M[x]^r)_{v_0}$ and the output of $\mathsf{H}_1$ is a sample from $(M'[x]^r)_{v_0}$, completing the proof.
	\end{proof}

	\begin{figure}
		\begin{framed}
			\begin{enumerate}
				\item Initialize $v = v_0$. Repeat $r$ times:
				\begin{enumerate}
					\item Pick $b \in [B]$ uniformly at random and let $I = I_b$.
					\item Pick $y'' \in \{0, 1\}^T$ uniformly at random.
					\item Let $v = \mathcal{P}\vert_{[n] \setminus I}(v; x, y'')$.
				\end{enumerate}
				\item Output $v$.
			\end{enumerate}
			\vspace{-5mm}
		\end{framed}
		\vspace{-5mm}
		\caption{The algorithm $\mathsf{H}_2$ defining the second hybrid distribution used to prove \cref{eqn:simulate-branching-program-correctness}. The only difference between $\mathsf{H}_1$ and $\mathsf{H}_2$ is that $\mathsf{H}_2$ feeds true randomness to $\mathcal{P}\vert_{[n] \setminus I}$, instead of feeding it a pseudorandom string from Nisan's generator.} \label{fig:hybrid-2}
	\end{figure}
	
	The second hybrid distribution is defined by the algorithm $\mathsf{H}_2$ given by \cref{fig:hybrid-2}.
	\begin{claim}[$\mathsf{H}_1 \approx \mathsf{H}_2$] \label{clm:hybrid-1-hybrid-2}
		For every $x$,
		\begin{equation}
			\mathsf{H}_1(\mathcal{P}, v_0, x, T) \sim_{\epsilon r} \mathsf{H}_2(\mathcal{P}, v_0, x, T).
		\end{equation}
	\end{claim}

	\begin{proof}
		This follows immediately from the correctness of $\mathsf{NisGen}$ and an application of \cref{lem:markov} that is perfectly analogous to the reasoning used to prove \cref{clm:a-hybrid-1}.
	\end{proof}

	\begin{figure}
		\begin{framed}
			\begin{enumerate}
				\item Initialize $v = v_0$. Repeat until $v$ is a terminal vertex of $\mathcal{P}$:
				\begin{enumerate}
					\item Pick $b \in [B]$ uniformly at random and let $I = I_b$.
					\item Pick $y'' \in \{0, 1\}^T$ uniformly at random.
					\item Let $v = \mathcal{P}\vert_{[n] \setminus I}(v; x, y'')$.
				\end{enumerate}
				\item Output $v$.
			\end{enumerate}
			\vspace{-5mm}
		\end{framed}
		\vspace{-5mm}
		\caption{The algorithm $\mathsf{H}_3$ defining the third hybrid distribution used to prove \cref{eqn:simulate-branching-program-correctness}. The only difference between $\mathsf{H}_2$ and $\mathsf{H}_3$ is that $\mathsf{H}_2$ terminates after $r$ iterations, whereas $\mathsf{H}_3$ waits until it reaches a terminal vertex of $\mathcal{P}$.} \label{fig:hybrid-3}
	\end{figure}

	Next, we must show that the output of $\mathsf{H}_2$ is statistically close to the output of $\mathsf{H}_3$. The idea is that in each iteration, with high probability, $\mathsf{H}_2$ progresses by roughly $B$ steps before running into a vertex $v$ with $i(v) \in I$. (Recall that $i(v)$ is the index of the input queried by vertex $v$.) Therefore, in total, with high probability, $\mathsf{H}_2$ progresses roughly $rB$ steps, which is at least $T$ by our choice of $r$. We now give the detailed statement and proof.
	
	\begin{claim}[$\mathsf{H}_2 \approx \mathsf{H}_3$] \label{clm:hybrid-2-hybrid-3}
		For every $x$,
		\begin{equation}
			\mathsf{H}_2(\mathcal{P}, v_0, x, T) \sim_{\exp(-r/8)} \mathsf{H}_3(\mathcal{P}, v_0, x, T).
		\end{equation}
	\end{claim}

	\begin{proof}
		Consider iteration $t$ of the loop in $\mathsf{H}_2$, where $1 \leq t \leq r$. Let $T_t$ be the number of steps through $\mathcal{P}\vert_{[n] \setminus I}$ that are taken in iteration $t$ when updating $v = \mathcal{P}\vert_{[n] \setminus I}(v; x, y'')$ before reaching a vertex that tries to query from $I$. (If we never reach such a vertex, i.e., we reach a terminal vertex of $\mathcal{P}$, then let $T_t = T$.) 
		We claim that
		\begin{equation} \label{eqn:game}
			\Pr\left[\sum_{t = 1}^r T_t < T\right] \leq e^{-r/8}.
		\end{equation}
		Proof: For $t \in [r]$, consider the value of $v$ at the beginning of iteration $t$ and the string $y'' \in \{0, 1\}^T$ chosen in iteration $t$. As a thought experiment, consider computing $\mathcal{P}(v; x, y'')$, i.e., taking a walk through the \emph{unrestricted} program. Let $v = u_0, u_1, u_2, \dots, u_{T'}$ be the vertices visited in this walk, $T' \leq T$.
		Let $S_t$ be the set of blocks $b' \in [B]$ that are queried by the first $B/2$ steps of this walk. That is,
		\begin{equation}
			S_t = \{b' \in [B]: \exists h < \lfloor B/2 \rfloor \text{ such that } i(u_h) \in I_{b'}\},
		\end{equation}
		so that $|S_t| \leq \lfloor B/2 \rfloor$. Let $S_t' = S_t \cup [B']$, where $B'$ is chosen so that $|S_t'| = \lfloor B/2 \rfloor$. Let $E_t$ be the event that $b \in S'_t$, where $b$ is the value chosen by $\mathsf{H}_2$ in iteration $t$ of the loop.
		
		Since $b$ and $y''$ are chosen \emph{independently} at random, the events $E_t$ are independent, and $\Pr[E_t] = \frac{\lfloor B/2 \rfloor}{B} \leq 1/2$. Therefore, by Hoeffding's inequality,
		\begin{equation}
			\Pr[\# E_t \text{ that occur} > (3/4)r] \leq e^{-r/8}.
		\end{equation}
		Now, suppose that $E_t$ does not occur. Then $b \not \in S_t'$, so $b \not \in S_t$. This implies that when updating $v = \mathcal{P}\vert_{[n] \setminus I}(v; x, y'')$ (taking a walk through the \emph{restricted} program), we either reach a terminal vertex of $\mathcal{P}$ or we take at least $\lfloor B/2 \rfloor$ steps before reaching a vertex that tries to query $I$. Therefore, $T_t \geq \min\{\lfloor B/2 \rfloor, T\}$. By \cref{eqn:r-def},
		\begin{align}
			\frac{r}{4} \cdot \left\lfloor \frac{B}{2} \right\rfloor \geq r \cdot \left(\frac{B}{8} - 1\right) \geq T.
		\end{align}
		\Cref{eqn:game} follows. Since $T \geq \length(\mathcal{P})$, $\sum_{t = 1}^r T_t \geq T$ implies that $\mathsf{H}_2$ outputs a terminal vertex of $\mathcal{P}$. Therefore, any random string that gives $\sum_{t = 1}^r T_t \geq T$ also causes $\mathsf{H}_2$ and $\mathsf{H}_3$ to output the same vertex.
	\end{proof}

	Finally, we argue that $\mathsf{H}_3$ perfectly simulates $\mathcal{P}$ (with zero error).
	\begin{claim}[$\mathsf{H}_3 \sim \mathcal{P}$] \label{clm:hybrid-3-p}
		For every $x$,
		\begin{equation}
			\mathsf{H}_3(\mathcal{P}, v_0, x, T) \sim \mathcal{P}(v_0; x, U_T).
		\end{equation}
	\end{claim}

	\begin{proof}
		For any path $v_0, v_1, \dots, v_{T'}$ through $\mathcal{P}$ ending at a terminal vertex, both computations, $\mathsf{H}_3(\mathcal{P}, v_0, x, T)$ and $\mathcal{P}(v_0; x, U_T)$, have exactly a $2^{-T'}$ chance of following that path.
	\end{proof}

	\begin{proof}[Proof of \cref{thm:simulate-branching-program}]
		By \cref{clm:a-hybrid-1,clm:hybrid-1-hybrid-2,clm:hybrid-2-hybrid-3,clm:hybrid-3-p} and the triangle inequality,
		\begin{equation}
			\density\{x \in \{0, 1\}^n : \mathsf{A}(\mathcal{P}, v_0, x, T) \not \sim_{\delta} \mathcal{P}(v_0; x, U_T)\} \leq 2^{-S^c},
		\end{equation}
		where $\delta = \epsilon r + \epsilon' r \cdot 2^{S - 1} + e^{-r/8}$. By our choice of $\epsilon$ (\cref{eqn:eps-def}), the first term is at most $e^{-cS}/4$. By our choice of $\epsilon'$ (\cref{eqn:eps'-def}), the second term is also at most $e^{-cS}/4$. By our choice of $r$ (\cref{eqn:r-def}), the third term is at most $e^{-cS}/2$. Therefore, $\delta \leq e^{-cS}$.
	\end{proof}
	
	\subsection{Main result: Derandomizing uniform random-access algorithms}
	
	\cref{thm:simulate-branching-program} immediately implies $\mathbf{BPTISP}(n \cdot \poly(S), S)$ can be simulated by a typically-correct algorithm that runs in time $n \cdot \poly(S)$ and space $O(S)$ that uses only $\poly(S)$ random bits.
	\begin{corollary} \label{cor:bptisp-poly-coins}
		Fix a function $S(n) \geq \log n$ that is constructible in time $n \cdot \poly(S)$ and space $O(S)$, and fix a constant $c \in \N$. For every language $L \in \mathbf{BPTISP}(n \cdot \poly(S), S)$, there is a randomized algorithm $\mathcal{A}$ running in time $n \cdot \poly(S)$ and space $O(S)$ that uses $\poly(S)$ random bits such that
		\begin{equation} \label{eqn:bptisp-poly-coins}
			\density\{x \in \{0, 1\}^n : \Pr[\mathcal{A}(x) \neq L(x)] > 2^{-S^c}\} \leq 2^{-S^c}.
		\end{equation}
	\end{corollary}

	\begin{proof}
		Let $\mathcal{B}$ be the algorithm witnessing $L \in \mathbf{BPTISP}(n \cdot \poly(S), S)$. Let $c'$ be a constant so that $\mathcal{B}$ runs in time $n \cdot S^{c'}$. For $n \in \N$, let $\mathcal{P}_n$ be a randomized branching program, where each vertex in $V(\mathcal{P}_n)$ describes a configuration of $\mathcal{B}$ with at most $S(n)$ symbols written on each tape. For each vertex $v \in V(\mathcal{P}_n)$, let $i(v)$ be the location of the input tape read head in the configuration described by $v$, and let $j(v)$ be the location of the random tape read head in the configuration described by $v$. The transitions of $\mathcal{P}_n$ correspond to the transitions of $\mathcal{B}$ in the obvious way.
		
		By construction, $\mathcal{P}_n$ is an R-OW branching program with size $2^{O(S)}$ and length at most $n \cdot S^{c'}$. Furthermore, given a vertex $v$, the neighborhood of $v$ can be computed in $\poly(S)$ time and $O(S)$ space, simply by consulting the transition function for $\mathcal{B}$.
		
		Given $x \in \{0, 1\}^n$, the algorithm $\mathcal{A}_0$ runs the algorithm of \cref{thm:simulate-branching-program} on input $(\mathcal{P}_n, v_0, x, n \cdot S^{c'})$, where $v_0$ encodes the starting configuration of $\mathcal{B}$. This gives a vertex $v \in V(\mathcal{P}_n)$. The algorithm $\mathcal{A}_0$ accepts if and only if $v$ encodes an accepting configuration of $\mathcal{B}$. That way,
		\begin{equation}
			\density\{x \in \{0, 1\}^n : \Pr[\mathcal{A}_0(x) \neq L(x)] > 1/3 + e^{-cS}\} \leq 2^{-S^{c}}.
		\end{equation}
		The algorithm $\mathcal{A}(x)$ runs $O(S^c)$ repetitions of $\mathcal{A}_0(x)$ and takes a majority vote, driving the failure probability down to $2^{-S^c}$.
		
		Clearly, $\mathcal{A}$ runs in time $n \cdot S^{c'} \cdot \poly(S) \cdot S^c = n \cdot \poly(S)$ and space $O(S)$. The number of random bits used by $\mathcal{A}$ is $O(\frac{n \cdot S^{c'}}{n} \cdot \poly(S) \cdot S^c) = \poly(S)$.
	\end{proof}
	
	We can further reduce the randomness complexity by using a pseudorandom generator by Nisan and Zuckerman \cite{nz96}.
	\begin{theorem}[\cite{nz96}] \label{thm:nz96}
		Fix constants $c \in \N, \alpha > 0$. For every $S \in \N$, there is a generator $\mathsf{NZGen}: \{0, 1\}^s \to \{0, 1\}^{S^c}$ with seed length $s \leq O(S)$ such that if $\mathcal{P}$ is an R-OW randomized branching program of size $2^S$, $v$ is a vertex, and $x$ is an input, then
		\begin{equation}
			\mathcal{P}(v; x, \mathsf{NZGen}(U_s)) \sim_{\epsilon} \mathcal{P}(v; x, U_{S^c}),
		\end{equation}
		where $\epsilon = 2^{-S^{1 - \alpha}}$. Given $S$ and $z$, $\mathsf{NZGen}(z)$ can be computed in $O(S)$ space and $\poly(S)$ time.
	\end{theorem}
	
	\begin{corollary}[Main result] \label{cor:bptisp-linear-coins}
		Fix a function $S(n) \geq \log n$ that is constructible in time $n \cdot \poly(S)$ and space $O(S)$, and fix constants $c \in \N, \alpha > 0$. For every language $L \in \mathbf{BPTISP}(n \cdot \poly(S), S)$, there is a randomized algorithm $\mathcal{A}$ running in time $n \cdot \poly(S)$ and space $O(S)$ that uses $O(S)$ random bits such that
		\begin{equation}
			\density\{x \in \{0, 1\}^n : \Pr[\mathcal{A}(x) \neq L(x)] > 2^{-S^{1 - \alpha}}\} \leq 2^{-S^c}.
		\end{equation}
	\end{corollary}

	\begin{proof}[Proof sketch]
		Compose the algorithm of \cref{cor:bptisp-poly-coins} with the Nisan-Zuckerman generator (\cref{thm:nz96}). The algorithm of \cref{cor:bptisp-poly-coins} can be implemented as a randomized branching program as in the proof of \cref{cor:bptisp-poly-coins}.
	\end{proof}

	Finally, we can eliminate the random bits entirely at the expense of time.
	
	\begin{corollary} \label{cor:bptisp-zero-coins}
		For every space-constructible function $S(n) \geq \log n$, for every constant $c \in \N$,
		\begin{equation}
			\mathbf{BPTISP}(n \cdot \poly(S), S) \text{ is within } 2^{-S^c} \text{ of } \mathbf{DSPACE}(S).
		\end{equation}
	\end{corollary}

	\begin{proof}
		Run the algorithm of \cref{cor:bptisp-linear-coins} on all possible random strings and take a majority vote.
	\end{proof}
	
	\section{Derandomizing Turing machines} \label{sec:bptisp-tm}
	
	In this section, we give our improved typically-correct derandomizations for Turing machines. \Cref{sec:s-ow,sec:bptisp-tm-1} concern derandomization with a low mistake rate, and \cref{sec:s-r,sec:amplification,sec:bptisp-tm-2} concern derandomization of Turing machines with runtime $T \approx n^2$.
	
	\subsection{Low-randomness simulation of sequential-access branching programs with a low mistake rate} \label{sec:s-ow}

	Recall that for a nonterminal vertex $v$ in a branching program, $i(v)$ is the index of the input queried by $v$, and $j(v)$ is the index of the random string queried by $v$.
	\begin{definition}
		An \emph{S-OW randomized branching program} is a randomized branching program $\mathcal{P}$ such that for every edge $(v, v')$ between two nonterminal vertices, $|i(v) - i(v')| \leq 1$ and $j(v') \in \{j(v), j(v) + 1\}$.
	\end{definition}
	In words, an S-OW randomized branching program has \emph{sequential} access to its input and one-way access to its random bits. By ``sequential access'', we mean that after reading bit $i$, it reads bit $i - 1$, bit $i$, or bit $i + 1$, like a head of a Turing machine. For S-OW branching programs, we give an algorithm analogous to \cref{thm:simulate-branching-program} but with a much lower rate of mistakes.
	\begin{theorem} \label{thm:simulate-s-ow-branching-program}
		For each constant $c \in \N$, there is a randomized random-access algorithm $\mathsf{A}$ with the following properties. Suppose $\mathcal{P}$ is an S-OW randomized branching program on $\{0, 1\}^n \times \{0, 1\}^T$ with $S \geq \log n$, where $S \stackrel{\normalfont \text{def}}{=} \lceil \log \size(\mathcal{P}) \rceil$. Suppose $v_0 \in V(\mathcal{P})$, $T \geq \length(\mathcal{P})$, and $x \in \{0, 1\}^n$. Then $\mathsf{A}(\mathcal{P}, v_0, x, T)$ outputs a vertex $v \in V(\mathcal{P})$. The number of random bits used by $\mathsf{A}$ is $\lceil T/n \rceil \cdot \poly(S)$, and $\mathsf{A}$ runs in time\footnote{Like in \cref{thm:simulate-branching-program}, the graph of $\mathcal{P}$ should be encoded in adjacency list format. We also stress that $\mathsf{A}$ is a random-access simulation of sequential-access branching programs.} $T \cdot \poly(n, S)$ and space $O(S)$. Finally, for every such $\mathcal{P}, v_0, T$,
		\begin{equation} \label{eqn:simulate-s-ow-branching-program}
			\#\{x \in \{0, 1\}^n : \mathsf{A}(\mathcal{P}, v_0, x, T) \not \sim_{\exp(-cS)} \mathcal{P}(v_0; x, U_T)\} \leq 2^{n/S^c}.
		\end{equation}
	\end{theorem}
	
	The proof of \cref{thm:simulate-s-ow-branching-program} is very similar to the proof of \cref{thm:simulate-branching-program}. The main difference is that instead of using a \emph{small} part of the input as the source of randomness, we use \emph{most} of the input as a source of randomness. The only part of the input that is \emph{not} used as a source of randomness is the region near the bit that the branching program was processing at the beginning of the current phase.
	
	Because the proof of \cref{thm:simulate-s-ow-branching-program} does not introduce any significantly new techniques, we defer the proof to \cref{apx:simulate-s-ow-branching-program}.
	
	\subsection{Derandomizing Turing machines with a low mistake rate} \label{sec:bptisp-tm-1}
	A \emph{randomized Turing machine} is defined like a randomized random-access Turing machine except that there are no index tapes. Thus, moving a read head from position $i$ to position $j$ takes $|i - j|$ steps. For functions $T, S: \N \to \N$, let $\mathbf{BPTISP}_{\text{TM}}(T, S)$ denote the class of languages $L$ such that there is a randomized Turing machine $\mathcal{A}$ that always runs in time $O(T(n))$ and space $O(S(n))$ such that for every $x \in \{0, 1\}^*$,
	\begin{equation}
		\Pr[\mathcal{A}(x) = L(x)] \geq 2/3.
	\end{equation}
	
	Trivially, a randomized Turing machine can be simulated by a randomized random-access Turing machine without loss in efficiency. Conversely, a single step of a randomized $O(S)$-space random-access Turing machine can be simulated in $O(n + S)$ steps by a randomized Turing machine. This proves the following elementary containments.
	\begin{proposition}
		For any functions $T, S: \N \to \N$ with $S(n) \geq \log n$,
		\begin{equation}
		\mathbf{BPTISP}_{\normalfont\text{TM}}(T, S) \subseteq \mathbf{BPTISP}(T, S) \subseteq \mathbf{BPTISP}_{\normalfont\text{TM}}(T \cdot (n + S), S).
		\end{equation}
	\end{proposition}
	
	\Cref{thm:simulate-s-ow-branching-program} combined with the Nisan-Zuckerman generator \cite{nz96} immediately implies a derandomization theorem for Turing machines analogous to \cref{cor:bptisp-linear-coins}.

	%

	\begin{corollary} \label{cor:bptisp-tm-1-linear-coins}
		Fix a function $S: \N \to \N$ with $S(n) \geq \log n$ that is constructible in time $\poly(n, S)$ and space $O(S)$, and fix constants $c \in \N, \alpha > 0$. For every language $L \in \mathbf{BPTISP}_{\normalfont\text{TM}}(n \cdot \poly(S), S)$, there is a randomized algorithm $\mathcal{A}$ running in time $\poly(n, S)$ and space $O(S)$ that uses $O(S)$ random bits such that
		\begin{equation}
			\#\{x \in \{0, 1\}^n : \Pr[\mathcal{A}(x) \neq L(x)] > 2^{-S^{1 - \alpha}}\} \leq 2^{n/S^c}.
		\end{equation}
	\end{corollary}

	\begin{proof}[Proof sketch]
		A randomized Turing machine obviously gives rise to an S-OW randomized branching program. Like in the proof of \cref{cor:bptisp-poly-coins} (but with \cref{thm:simulate-s-ow-branching-program} in place of \cref{thm:simulate-branching-program}), we first obtain an algorithm that uses $\poly(S)$ random bits. Composing with the Nisan-Zuckerman generator (\cref{thm:nz96}) completes the proof.
	\end{proof}

	\begin{corollary} \label{cor:bptisp-tm-1-zero-coins}
		For every space-constructible function $S(n) \geq \log n$, for every constant $c \in \N$,
		\begin{equation}
			\mathbf{BPTISP}_{\normalfont\text{TM}}(n \cdot \poly(S), S) \text{ is within } 2^{-n + n/S^c} \text{ of } \mathbf{DSPACE}(S).
		\end{equation}
	\end{corollary}

	\begin{proof}
		Simulate the algorithm of \cref{cor:bptisp-tm-1-linear-coins} on all possible random strings and take a majority vote.
	\end{proof}
	
	\subsection{Simulating branching programs with random access to random bits} \label{sec:s-r}
	
	We now move on to our second derandomization of Turing machines, as outlined in \cref{sec:tm-overview}. Recall that for a nonterminal vertex $v$ in a branching program, $i(v)$ is the index of the input that is queried by $v$.
	
	\begin{definition}
		An \emph{S-R} randomized branching program is a randomized branching program $\mathcal{P}$ such that for every edge $(v, v')$ between two nonterminal vertices, $|i(v) - i(v')| \leq 1$.
	\end{definition}

	In words, an S-R randomized branching program has sequential access to its input and random access to its random bits. This model is \emph{more general} than the S-OW model; the S-OW model corresponds more directly to the randomized Turing machine model. But studying the more general S-R model will help us derandomize Turing machines.
	
	We will give a randomness-efficient algorithm for simulating S-R randomized branching programs, roughly analogous to \cref{thm:simulate-branching-program,thm:simulate-s-ow-branching-program}. The simulation will only work well if the branching program has small length \emph{and} uses few random bits.
	
	Our simulation of S-R randomized branching programs is a fairly straightforward application of work by Kinne et al.\ \cite{kvms12}; this section is not technically novel. But it is useful to be able to compare the work by Kinne et al.\ \cite{kvms12} to our algorithms based on the ``out of sight, out of mind'' technique.

	Unlike \cref{thm:simulate-branching-program,thm:simulate-s-ow-branching-program}, our simulation of S-R branching programs will not work on a step-by-step basis, generating a distribution on vertices that approximates the behavior of the branching program. Instead, our simulation of S-R branching programs will only work for S-R branching programs that \emph{compute a Boolean function}. We now give the relevant definition.
	
	\begin{definition}
		Let $\mathcal{P}$ be a randomized branching program on $\{0, 1\}^n \times \{0, 1\}^m$. Suppose some vertex $v_0 \in V(\mathcal{P})$ is labeled as the \emph{start vertex}, and every terminal vertex of $\mathcal{P}$ is labeled with an \emph{output bit} $b \in \{0, 1\}$. In this case, we identify $\mathcal{P}$ with a function $\mathcal{P}: \{0, 1\}^n \times \{0, 1\}^m \to \{0, 1\}$ defined by
		\begin{equation}
			\mathcal{P}(x, y) = \text{the output bit labeling } \mathcal{P}(v_0; x, y).
		\end{equation}
		We say that \emph{$\mathcal{P}$ computes $f: \{0, 1\}^n \to \{0, 1\}$ with failure probability $\delta$} if for every $x \in \{0, 1\}^n$,
		\begin{equation}
			\Pr[\mathcal{P}(x, U_m) = f(x)] \geq 1 - \delta.
		\end{equation}
	\end{definition}
	
	Instead of assuming a time bound, it will be useful to assume a bound on the \emph{query complexity} of the branching program.
	\begin{definition}
		Let $\mathcal{P}$ be randomized branching program. The \emph{query complexity} of $\mathcal{P}$, denoted $\queries(\mathcal{P})$, is the maximum, over all paths $v_1, v_2, \dots, v_T$ through $\mathcal{P}$ consisting entirely of nonterminal vertices, of
		\begin{equation}
			1 + \#\{t \in \{2, 3, \dots, T\} : i(v_t) \neq i(v_{t - 1})\}.
		\end{equation}
	\end{definition}
	In words, $\queries(\mathcal{P})$ is the number of steps that $\mathcal{P}$ takes in which it queries a \emph{new} bit of its input, i.e., not the bit that it queried in the previous step. Trivially, $\queries(\mathcal{P}) \leq \length(\mathcal{P})$. The reader is encouraged to think of the distinction between $\queries(\mathcal{P})$ and $\length(\mathcal{P})$ as being a technicality that can be ignored on the first reading.
	
	We can now state our deterministic simulation theorem for S-R randomized branching programs. It consists of a method of deterministically generating coins for the branching program from its input.
	
	\begin{theorem} \label{thm:simulate-s-r-branching-program}
		There is a constant $\alpha > 0$ so that for every $n, m$ with $m \leq n/3$, there is a function $\mathsf{R}: \{0, 1\}^n \to \{0, 1\}^m$ with the following properties. Suppose $\mathcal{P}$ is an S-R randomized branching program on $\{0, 1\}^n \times \{0, 1\}^m$ that computes a function $f$ with failure probability $\delta$. Suppose $TSm \leq \alpha n^2$, where $T \stackrel{\normalfont\text{def}}{=} \queries(\mathcal{P})$ and $S \stackrel{\normalfont\text{def}}{=} \lceil \log \size(\mathcal{P}) \rceil$. Then
		\begin{equation}
			\density\{x \in \{0, 1\}^n : \mathcal{P}(x, \mathsf{R}(x)) \neq f(x)\} \leq 3\delta + m \cdot 2^{-\alpha n / m}.
		\end{equation}
		Furthermore, given $x$ and $m$, $\mathsf{R}(x)$ can be computed in space $O(\log n)$.
	\end{theorem}

	The function $\mathsf{R}$ is based on a pseudorandom generator by Kinne et al.\ \cite{kvms12} for \emph{multiparty communication protocols}. In a \emph{public-coin randomized $3$-party NOF protocol} $\Pi$, there are three parties, three inputs $x_1, x_2, x_3$, and one random string $y$. Party $i$ knows $x_j$ for $j \neq i$, and all three parties know $y$. All parties have access to a blackboard. The protocol specifies who should write next as a function of what has been written on the blackboard so far and $y$. Eventually, the protocol specifies the output $\Pi(x_1, x_2, x_3, y)$, which should be a function of what has been written on the blackboard and $y$. The communication complexity of $\Pi$ is the maximum number of bits written on the blackboard over all $x_1, x_2, x_3, y$. A \emph{deterministic $3$-party NOF protocol} is just the case $|y| = 0$.
	
	Following Kinne et al.\ \cite{kvms12}, we rely on a $3$-party communication complexity lower bound by Babai et al.\ \cite{bns92}. For an integer $\ell \in \N$, define $\gip_{\ell}: (\{0, 1\}^{\ell})^3 \to \{0, 1\}$ to be the generalized inner product function, i.e.,
	\begin{equation}
	\gip_{\ell}(x, y, z) = \sum_{i = 1}^{\ell} x_i y_i z_i \text{ mod } 2.
	\end{equation}
	Babai et al.\ showed that the trivial communication protocol for $\gip_{\ell}$ is essentially optimal, even in the average-case setting.
	\begin{theorem}[\cite{bns92}] \label{thm:bns92}
		There is a constant $\beta > 0$ so that for every $\ell \in \N, \epsilon > 0$, if $\Pi$ is a deterministic $3$-party NOF protocol with
		\begin{equation}
		\Pr_{x, y, z}[\Pi(x, y, z) = \normalfont\text{GIP}_{\ell}(x, y, z)] \geq \frac{1}{2} + \epsilon,
		\end{equation}
		then the communication complexity of $\Pi$ is at least $\beta \cdot(\ell - \log(1/\epsilon))$.
	\end{theorem}
	
	To define $\mathsf{R}$, let $x \in \{0, 1\}^n$. Partition $n = n_1 + n_2 + n_3$, where $n_i \geq \lfloor n/3 \rfloor$ for each $i$. Correspondingly partition $x = x_1 \circ x_2 \circ x_3$, where $|x_i| = n_i$. Define
	\begin{equation}
		\ell = \left\lfloor\frac{\lfloor n/3 \rfloor}{m}\right\rfloor,
	\end{equation}
	so that $\ell \geq 1$. For $i \in [3]$ and $j \in [m]$, let $x_{ij}$ be the $j$th $\ell$-bit substring of $x_i$. (Note that due to roundoff errors, for some values of $n$, some bits of $x$ are not represented in any $x_{ij}$.) Then we define
	\begin{equation}
		\mathsf{R}(x) = \text{GIP}_{\ell}(x_{11}, x_{21}, x_{31}) \circ \dots \circ \text{GIP}_{\ell}(x_{1m}, x_{2m}, x_{3m}) \in \{0, 1\}^m.
	\end{equation}
	
	Kinne et al.\ observed that $x \mapsto (x, \mathsf{R}(x))$ is a pseudorandom generator that fools $3$-party NOF protocols \cite{kvms12}. For clarity, we reproduce the argument here.
	
	\begin{lemma} \label{lem:gen-fools-cc}
		Suppose $\Pi: \{0, 1\}^{n_1} \times \{0, 1\}^{n_2} \times \{0, 1\}^{n_3} \times \{0, 1\}^m \to \{0, 1\}$ is a public-coin randomized $3$-party NOF protocol. Suppose that for some $\epsilon > 0$, $\Pi$ uses less than $\beta \cdot (\ell - \log(1/\epsilon))$ bits of communication, where $\beta$ is the constant of \cref{thm:bns92}. Then
		\begin{equation}
			\left|\Pr_{x_1, x_2, x_3, y}[\Pi(x_1, x_2, x_3, y) = 1] - \Pr_{x_1, x_2, x_3}[\Pi(x_1, x_2, x_3, \mathsf{R}(x_1, x_2, x_3)) = 1]\right| < \epsilon m.
		\end{equation}
	\end{lemma}

	\begin{proof}
		Let
		\begin{equation}
			\delta = \left|\Pr_{x_1, x_2, x_3, y}[\Pi(x_1, x_2, x_3, y) = 1] - \Pr_{x_1, x_2, x_3}[\Pi(x_1, x_2, x_3, \mathsf{R}(x_1, x_2, x_3)) = 1]\right|.
		\end{equation}
		By Yao's distinguisher-to-predictor argument \cite{yao82}, there is some index $i \in [m]$ and a protocol $\Pi': \{0, 1\}^{n_1} \times \{0, 1\}^{n_2} \times \{0, 1\}^{n_3} \times \{0, 1\}^{i - 1} \to \{0, 1\}$ so that
		\begin{equation}
			\Pr_{x_1, x_2, x_3}[\Pi'(x_1, x_2, x_3, \mathsf{R}(x_1, x_2, x_3) \vert_{[i - 1]}) = \mathsf{R}(x_1, x_2, x_3)_i] \geq \frac{1}{2} + \frac{\delta}{m}.
		\end{equation}
		The protocol $\Pi'$ is a public-coin randomized $3$-party NOF protocol that still uses less than $\beta \cdot (\ell - \log(1/\epsilon))$ bits of communication, since it merely involves simulating $\Pi$ with certain input/coin bits fixed to certain values and possibly negating the output. This immediately implies a protocol for $\gip_{\ell}$ with the same parameters with advantage $\delta/m$. There is some way to fix the randomness to preserve advantage, so by \cref{thm:bns92}, $\delta/m < \epsilon$.
	\end{proof}
	
	The connection between S-R randomized branching programs and 3-party communication protocols is given by the following lemma.
	
	\begin{lemma} \label{lem:branching-program-cc}
		There is a public-coin randomized $3$-party NOF protocol $\Pi: \{0, 1\}^{n_1} \times \{0, 1\}^{n_2} \times \{0, 1\}^{n_3} \times \{0, 1\}^m \to \{0, 1\}$ such that
		\begin{equation}
			\Pi(x_1, x_2, x_3, y) = \mathcal{P}(x_1 \circ x_2 \circ x_3, y),
		\end{equation}
		and $\Pi$ uses only $O(\frac{TS}{n})$ bits of communication.
	\end{lemma}

	\begin{proof}
		Parties $1$ and $3$ alternate simulating the operation of $\mathcal{P}$. If party $1$ is simulating and the program reads from the first $n_1$ bits of the input, party $1$ sends the state to party $3$. Similarly, if party $3$ is simulating and the program reads from the last $n_3$ bits of the input, party $3$ sends the state to party $1$. Each such transition indicates that the program must have spent at least $n_2$ steps traversing the middle $n_2$ bits of the input. Therefore, the total number of such transitions is at most $\frac{T}{n_2}$.
	\end{proof}

	Given \cref{lem:gen-fools-cc,lem:branching-program-cc}, \cref{thm:simulate-s-r-branching-program} follows by a lemma by Kinne et al.\ \cite[Lemma 1]{kvms12}. For clarity, we reproduce the argument here.

	\begin{proof}[Proof of \cref{thm:simulate-s-r-branching-program}]
		The best case is at least as good as the average case, so there is some string $y_* \in \{0, 1\}^m$ such that
		\begin{equation}
			\Pr_{x \in \{0, 1\}^n}[\mathcal{P}(x, y_*) \neq f(x)] \leq \delta.
		\end{equation}
		Define $g: \{0, 1\}^n \times \{0, 1\}^m \to \{0, 1\}$ by
		\begin{equation}
			g(x, y) = \begin{cases}
				1 & \text{if } \mathcal{P}(x, y) = \mathcal{P}(x, y_*) \\
				0 & \text{otherwise.}
			\end{cases}
		\end{equation}
		Think of $x \in \{0, 1\}^n$ as $x = x_1 \circ x_2 \circ x_3$, like in the definition of $\mathsf{R}$. Then by \cref{lem:branching-program-cc}, $g$ can be computed by a $3$-party NOF protocol using $O(\frac{TS}{n})$ bits of communication. By choosing $\alpha$ small enough and setting $\epsilon = 2^{-\alpha n/m}$, this protocol for $f$ will use fewer than $\beta (\ell - \log(1/\epsilon))$ bits of communication. Therefore, by \cref{lem:gen-fools-cc},
		\begin{equation}
			\Pr_x[\mathcal{P}(x, \mathsf{R}(x)) \neq \mathcal{P}(x, y_*)] \leq \Pr_{x, y}[\mathcal{P}(x, y) \neq \mathcal{P}(x, y_*)] + \epsilon m.
		\end{equation}
		Therefore,
		\begin{align}
			\Pr_x[\mathcal{P}(x, \mathsf{R}(x)) \neq f(x)] &\leq \Pr_x[\mathcal{P}(x, y_*) \neq f(x)] + \Pr_x[\mathcal{P}(x, \mathsf{R}(x)) \neq \mathcal{P}(x, y_*)] \\
			&\leq \delta + \Pr_{x, y}[\mathcal{P}(x, y) \neq \mathcal{P}(x, y_*)] + \epsilon m \\
			&\leq \delta + \Pr_{x, y}[\mathcal{P}(x, y) \neq f(x)] + \Pr_x[\mathcal{P}(x, y_*) \neq f(x)] + \epsilon m \\
			&\leq \delta + \delta + \delta + \epsilon m.
		\end{align}
		Obviously, $\mathsf{R}(x)$ can be computed in $O(\log n)$ space.
	\end{proof}
	
	\subsection{Randomness-efficient amplification for branching programs} \label{sec:amplification}
	
	We will use a space-efficient expander walk algorithm by Gutfreund and Viola \cite{gv04}.
	\begin{theorem}[\cite{gv04}] \label{thm:gv04}
		For every $s \in \N$, there is a constant-degree expander graph $G$ on vertex set $\{0, 1\}^s$. Furthermore, there is an algorithm $\mathsf{GVWalk}$ such that if $y \in \{0, 1\}^s$ is a vertex and $e_1, e_2, \dots, e_r \in \{0, 1\}^{O(1)}$ are edge labels, then $\mathsf{GVWalk}(y, e_1, e_2, \dots, e_r)$ outputs the vertex reached by starting at $y$ and taking a walk by following the edge labels $e_1, e_2, \dots, e_r$. The algorithm $\mathsf{GVWalk}$ runs in space $O(\log s + \log r)$.
	\end{theorem}
	
	Recall that we are working toward derandomizing the class $\mathbf{BPTISP}_\text{TM}(T, S)$ for all $TS^2 \leq o(n^2 / \log n)$. This class corresponds to branching programs on $\{0, 1\}^n \times \{0, 1\}^T$ that compute some function with failure probability $1/3$. But \cref{thm:simulate-s-r-branching-program} requires that the branching program use at most $\frac{\alpha n^2}{TS}$ random bits. Furthermore, the failure probability of the branching program governs the mistake rate of the derandomization.
	
	We can overcome these two difficulties because randomized Turing machines correspond to \emph{S-OW} randomized branching programs (i.e., programs that have sequential access to the input and one-way access to the random bits), whereas \cref{thm:simulate-s-r-branching-program} applies to the more powerful S-R model (i.e., programs that have sequential access to the input and \emph{random} access to the random bits). An S-OW branching program can be simulated by an S-R branching program using very few random bits by applying Nisan's generator. The following lemma combines this idea with a random walk on an expander graph (\cref{thm:gv04}) for amplification. This is the same technique that Fortnow and Klivans used to prove that $\mathbf{BPL} \subseteq \mathbf{L}/O(n)$ \cite{fk06}.
	
	\begin{lemma} \label{lem:s-ow-to-s-r}
		Suppose $\mathcal{P}$ is an S-OW randomized branching program on $\{0, 1\}^n \times \{0, 1\}^T$ that computes a function $f: \{0, 1\}^n \to \{0, 1\}$ with failure probability $1/3$. Let $S = \log \size(\mathcal{P})$. For every $\delta > 0$, there is an S-R branching program $\mathcal{P}'$ on $\{0, 1\}^n \times \{0, 1\}^m$ that computes $f$ with failure probability $\delta$ such that
		\begin{align}
			\queries(\mathcal{P}') &\leq O((\queries(\mathcal{P}) + n) \log(1/\delta)), \\
			\log \size(\mathcal{P}') &\leq O(S + \log \log(1/\delta)), \\
			m &\leq O(S \log T + \log(1/\delta)).
		\end{align}
		Furthermore, given $\mathcal{P}$, $\delta$, and a vertex $v \in V(\mathcal{P}')$, the neighborhood of $v$ can be computed in time\footnote{As usual, we assume that the graph of $\mathcal{P}$ is encoded in adjacency list format. We also assume that the start vertex $v_0$ is designated in a way that allows it to be computed in the specified time and space.} $\poly(S, \log(1/\delta))$ and space $O(S + \log \log(1/\delta))$.
	\end{lemma}

	\begin{proof}
		Let $\mathsf{NisGen}: \{0, 1\}^s \to \{0, 1\}^T$ be Nisan's generator with error $0.1$ for randomized branching programs of size $\size(\mathcal{P})$. Let $G$ be the expander of \cref{thm:gv04} on vertex set $\{0, 1\}^s$. We will interpret a string $y \in \{0, 1\}^m$ as describing a walk through $G$ from an arbitrary initial vertex of length $r - 1$, so that $m = s + O(r)$. Let $y_1, \dots, y_r \in \{0, 1\}^s$ be the vertices visited by this walk. The program $\mathcal{P}'(x, y)$ runs $\mathcal{P}(x, \mathsf{NisGen}(y_t))$ for every $y \in [r]$ and takes a majority vote of the answers; it finds the vertices $y_t$ by running the algorithm $\mathsf{GVWalk}$ of \cref{thm:gv04}. By the expander walk Chernoff bound \cite{gil98}, for an appropriate choice of $r = \Theta(\log(1/\delta))$, the failure probability of $\mathcal{P}'$ is at most $\delta$.
		
		Clearly, $\queries(\mathcal{P}') \leq r \cdot (\queries(\mathcal{P}) + n)$, where the $+ n$ term takes care of the steps needed to get from the final position of $x$ read in one iteration of $\mathcal{P}$ to the first position of $x$ read in the next iteration of $\mathcal{P}$ (recall that $\mathcal{P}'$ is an S-R branching program).
		
		The space needed by $\mathcal{P}'$ consists of the $S$ bits of space needed for $\mathcal{P}$, plus $O(S)$ bits of space for computing $\mathsf{NisGen}$, plus $O(\log r)$ bits of space to keep track of the answers generated by the iterations, plus $O(\log S + \log r)$ bits of space for $\mathsf{GVWalk}$. Finally, computing the neighborhood of $v$ merely requires inspecting the transition functions for the algorithms $\mathsf{NisGen}$ and $\mathsf{GVWalk}$, inspecting $\mathcal{P}$, and doing arithmetic.
	\end{proof}
	
	\subsection{Derandomizing Turing machines with runtime near $n^2$} \label{sec:bptisp-tm-2}
	Finally, we are ready to state and prove our typically-correct derandomization of $\mathbf{BPTISP}_{\text{TM}}(T, S)$ based on \cref{thm:simulate-s-r-branching-program}.
	
	\begin{corollary} \label{cor:bptisp-tm-2}
		Suppose $T, S: \N \to \N$ are both constructible in time $\poly(n)$ and space $O(S)$ and $TS^2 \leq o\left(\frac{n^2}{\log n}\right)$. For every language $L \in \mathbf{BPTISP}_{\normalfont\text{TM}}(T, S)$, there is a constant $\gamma > 0$ so that
		\begin{equation} \label{eqn:cor-tm-2}
			L \text{ is within } \exp\left(-\frac{\gamma n}{\sqrt{T S}}\right) + \exp\left(-\frac{\gamma n^2}{TS^2 \log n}\right) \text{ of } \mathbf{DTISP}(\poly(n), S).
		\end{equation}
	\end{corollary}

	The rate of mistakes in \cref{cor:bptisp-tm-2} is always $o(1)$. The rate of mistakes gets smaller (i.e., the simulation quality gets higher) when $T$ and $S$ are smaller. For example, if $S = \log n$ and $T = n^2/\log^4 n$, the rate of mistakes in \cref{eqn:cor-tm-2} is $n^{-\Omega(1)}$. For another example, if $S = \polylog n$ and $T = n \polylog n$, the rate of mistakes in \cref{eqn:cor-tm-2} is $\exp\left(-\tilde{\Omega}(\sqrt{n})\right)$.
	
	As a reminder, \cref{cor:bptisp-tm-2} is incomparable to \cref{cor:bptisp-zero-coins}: the randomized classes in the two results are incomparable; the deterministic algorithm in \cref{cor:bptisp-tm-2} is faster; the mistake rate in \cref{cor:bptisp-tm-2} is lower when $S$ and $T$ are not too big. Similarly, \cref{cor:bptisp-tm-2} is incomparable to \cref{cor:bptisp-tm-1-zero-coins}: the randomized class in \cref{cor:bptisp-tm-2} is more powerful and the deterministic algorithm in \cref{cor:bptisp-tm-2} is faster, but the mistake rate in \cref{cor:bptisp-tm-2} is much higher. Finally, even when $S \geq n^{\Omega(1)}$, \cref{cor:bptisp-tm-2} is incomparable to derandomizing via the Nisan-Zuckerman generator \cite{nz96}, because the deterministic algorithm of \cref{cor:bptisp-tm-2} runs in polynomial time, although it makes some mistakes.
	
	Conceptually, the proof of \cref{cor:bptisp-tm-2} merely consists of combining \cref{lem:s-ow-to-s-r,thm:simulate-s-r-branching-program}. The only work to be done is in appropriately choosing $\delta$ and verifying parameters.

	\begin{proof}[Proof of \cref{cor:bptisp-tm-2}]
		Let $\mathcal{A}$ be the algorithm witnessing $L \in \mathbf{BPTISP}_{\text{TM}}(T, S)$. Let $\mathcal{P}_n$ be the S-OW branching program on $\{0, 1\}^n \times \{0, 1\}^T$ describing the behavior of $\mathcal{A}$ on inputs of length $n$.
		
		We consider two cases. First, suppose $TS^3 > n^2/\log^2 n$. Then let
		\begin{equation}
			\delta = \exp\left(-\frac{\gamma_0 n^2}{TS^2 \log n}\right),
		\end{equation}
		where the constant $\gamma_0$ will be specified later. Let $\mathcal{P}'_n$ be the S-R branching program on $\{0, 1\}^n \times \{0, 1\}^m$ given by \cref{lem:s-ow-to-s-r}. There is a constant $c$ that does not depend on $\gamma$ so that
		\begin{align}
			\queries(\mathcal{P}'_n) \cdot \log \size(\mathcal{P}'_n) \cdot m &\leq c TS^2 \log n \ln(1/\delta) + cTS\ln^2(1/\delta) \\
			&= c \gamma_0 n^2 + \frac{c \gamma_0^2 n^4}{TS^3 \log^2 n} \\
			&\leq c \gamma_0 n^2 + c \gamma_0^2 n^2.
		\end{align}
		Choose $\gamma_0$ so that $c\gamma_0 + c \gamma_0^2 \leq \alpha$, where $\alpha$ is the value in \cref{thm:simulate-s-r-branching-program}. Since $TS^2 \leq o(n^2/\log n)$ and $T \geq n$, we must have $S \leq o(\sqrt{n / \log n})$. Therefore,
		\begin{equation}
			m \leq O\left(S \log n + \frac{n^2}{TS^2 \log n}\right) \leq O\left(S \log n + \frac{TS^3 \log n}{TS^2}\right) \leq o(\sqrt{n \log n}) \leq n/3.
		\end{equation}
		Therefore, the hypotheses of \cref{thm:simulate-s-r-branching-program} are satisfied.
		
		The deterministic algorithm, naturally, outputs $\mathcal{P}'_n(x, \mathsf{R}(x))$, where $\mathsf{R}$ is the function of \cref{thm:simulate-s-r-branching-program}. It is immediate that this runs in $\poly(n)$ time and $O(S)$ space. Finally, to compute the rate of mistakes, observe that
		\begin{align}
			m \cdot 2^{-\alpha n/m} \leq \exp\left(-\Omega\left(-\frac{n}{S \log n}\right)\right),
		\end{align}
		whereas
		\begin{align}
			\delta &\geq \exp\left(-O\left(\frac{n}{S^2 \log n}\right)\right).
		\end{align}
		Therefore, when $n$ is sufficiently large, $m \cdot 2^{-\alpha n/m} < \delta$. Therefore,
		\begin{equation}
			\density\{x \in \{0, 1\}^n : \mathcal{P}'_n(x, \mathsf{R}(x)) \neq L(x)\} \leq 4\delta.
		\end{equation}
		
		For the second case, suppose $TS^3 \leq n^2/\log^2 n$. Then let
		\begin{equation}
			\delta = \exp\left(-\frac{\gamma_0 n}{\sqrt{TS}}\right).
		\end{equation}
		Again, let $\mathcal{P}'_n$ be the S-R branching program on $\{0, 1\}^n \times \{0, 1\}^m$ given by \cref{lem:s-ow-to-s-r}. Then
		\begin{align}
			\queries(\mathcal{P}'_n) \cdot \log \size(\mathcal{P}'_n) \cdot m &\leq cTS^2 \log n \ln(1/\delta) + c TS \ln^2(1/\delta) \\
			&=c \gamma_1 \sqrt{TS^3} n \log n + c \gamma_1^2 n^2 \\
			&\leq c \gamma_1 n^2 + c \gamma_1^2 n^2 \\
			&\leq \alpha n^2.
		\end{align}
		Furthermore, since $TS^3 \leq n^2 / \log^2 n$, taking a square root gives $S \sqrt{TS} \leq n/\log n$, and hence
		\begin{equation}
			m \leq O\left(S \log n + \frac{n}{\sqrt{TS}}\right) \leq O\left(\frac{n}{\sqrt{TS}}\right) < n/3.
		\end{equation}
		Therefore, again, the hypotheses of \cref{thm:simulate-s-r-branching-program}. In this case as well, the deterministic algorithm outputs $\mathcal{P}'_n(x, \mathsf{R}(x))$. We now compute the rate of mistakes again. We have
		\begin{align}
			m \cdot 2^{-\alpha n/m} \leq \exp(-\Omega(\sqrt{TS})) < \delta
		\end{align}
		for sufficiently large $n$, because $\sqrt{TS} \geq \sqrt{n \log n}$. Therefore, once again,
		\begin{equation}
			\density\{x \in \{0, 1\}^n : \mathcal{P}'_n(x, \mathsf{R}(x)) \neq L(x)\} \leq 4\delta.
		\end{equation}
		
		Choosing $\gamma < \gamma_0$ completes the proof.
	\end{proof}

	\section{Derandomization with advice} \label{sec:bpl-advice}
	
	As previously mentioned, Fortnow and Klivans showed that $\mathbf{BPL} \subseteq \mathbf{L}/O(n)$ \cite{fk06}. We now explain how to refine their ideas and slightly improve their result. Fortnow and Klivans' argument relied on the Gutfreund-Viola space-efficient expander walk (\cref{thm:gv04}). They only used this expander for its sampling properties. Extractors also have good sampling properties. Our improvement will come from simply replacing the expander-based sampler in Fortnow and Klivans' argument with the GUV-based extractor of \cref{thm:guv09}.
	
	\begin{theorem} \label{thm:fk06-improved}
		$\mathbf{BPL} \subseteq \mathbf{L}/(n + O(\log^2 n))$.
	\end{theorem}

	\begin{proof}
		Let $\mathcal{A}$ be an algorithm witnessing $L \in \mathbf{BPL}$, and assume $\mathcal{A}$ has failure probability at most $0.1$. Let $\mathsf{NisGen}: \{0, 1\}^s \to \{0, 1\}^{\poly(n)}$ be Nisan's generator (\cref{thm:nis92}) with error $0.1$ and space bound sufficient to fool $\mathcal{A}$, so that $s \leq O(\log^2 n)$. Let $\mathsf{GUVExt}: \{0, 1\}^{n + 2s + 3} \times \{0, 1\}^d \to \{0, 1\}^s$ be the $(2s, 0.1)$-extractor of \cref{thm:guv09}, so that $d \leq O(\log n)$.
		
		Given input $x \in \{0, 1\}^n$ and advice $a \in \{0, 1\}^{n + 2s + 3}$, run $\mathcal{A}(x, \mathsf{NisGen}(\mathsf{GUVExt}(a, z)))$ for all $z$ and take a majority vote.
		
		This algorithm clearly runs in space $O(\log n)$. By \cref{prop:zuc97}, for each fixed $x$, the number of advice strings $a$ causing the algorithm to give the wrong answer is at most $2^{2s + 2}$. Therefore, the total number of advice strings $a$ that cause the algorithm to give the wrong answer for any $x$ is at most $2^{n + 2s + 2} < 2^{|a|}$. Therefore, there is some choice of $a$ such that the algorithm succeeds on all inputs.
	\end{proof}

	We now generalize \cref{thm:fk06-improved}, showing that the amount of advice can be reduced to below $n$ in certain cases. We will rely on a special feature of Nisan's generator that Nisan used to prove $\mathbf{RL} \subseteq \mathbf{SC}$. The seed to Nisan's generator is naturally divided into two parts, $s = s_1 + s_2$, where $s_2 \leq O(S + \log(1/\epsilon))$.\footnote{The first $s_1$ bits specify the hash functions, and the last $s_2$ bits specify the input to those hash functions.} Nisan showed that there is an efficient procedure to \emph{check} that the first part of the seed is ``good'' for a particular randomized log-space algorithm and a particular input to that algorithm.
	\begin{lemma}[\cite{nis94}] \label{lem:nis94}
		For every $S \in \N$, there is a function $\mathsf{NisGen}: \{0, 1\}^{s_1} \times \{0, 1\}^{s_2} \to \{0, 1\}^{2^S}$, with $s_1 \leq O(S^2)$ and $s_2 \leq O(S)$, and an algorithm $\mathsf{Check}$, so that
		\begin{itemize}
			\item For any R-OW randomized branching program $\mathcal{P}$ with $\log \size(\mathcal{P}) \leq S$ and any input $x \in \{0, 1\}^n$,
			\begin{equation}
				\Pr_{y_1 \in \{0, 1\}^{s_1}}[\mathsf{Check}(\mathcal{P}, x, y_1) = 1] \geq 1/2.
			\end{equation}
			\item If $\mathsf{Check}(\mathcal{P}, x, y_1) = 1$, then for any vertex $v_0 \in V(\mathcal{P})$,
			\begin{equation}
				\mathcal{P}(v_0; x, \mathsf{NisGen}(y_1, U_{s_2})) \sim_{0.1} \mathcal{P}(v_0; x, U_{2^S}).
			\end{equation}
		\end{itemize}
		Furthermore, $\mathsf{Check}$ runs in space $O(S)$, and given $S$, $y_1$, and $y_2$, $\mathsf{NisGen}(y_1, y_2)$ can be computed in space $O(S)$.
	\end{lemma}
	
	A \emph{$\mathbf{ZP \cdot SPACE}(S)$ algorithm} for a language $L$ with failure probability $\delta$ is a randomized Turing machine $\mathcal{A}$ with \emph{two-way} access to its random bits such that $\mathcal{A}$ runs in space $O(S)$, $\Pr[\mathcal{A}(x) \in \{L(x), \bot\}] = 1$, and $\Pr[\mathcal{A}(x) = \bot] \leq \delta$. The following lemma refines a theorem by Nisan that says that $\mathbf{BPL} \subseteq \mathbf{ZP \cdot L}$ \cite{nis93}; the improvement is that our algorithm has a low failure probability relative to the number of random bits it uses.
	\begin{lemma} \label{lem:zp.space}
		Fix $S: \N \to \N$ with $S(n) \geq \log n$ and $\delta: \N \to [0, 1]$, both constructible in space $O(S)$. For every $L \in \mathbf{BPSPACE}(S)$, there is a $\mathbf{ZP\cdot SPACE}(S)$ algorithm $\mathcal{A}$ that decides $L$ with failure probability $\delta$ and uses $\log_2(1/\delta) + O(S^2)$ random bits.
	\end{lemma}

	\begin{proof}
		Let $\mathcal{B}$ be the algorithm witnessing $L \in \mathbf{BPSPACE}(S)$, and assume $\mathcal{B}$ has failure probability at most $0.1$. Let $\mathcal{P}$ be the corresponding R-OW branching program for inputs of length $n$. Let $\mathsf{NisGen}: \{0, 1\}^{s_1} \times \{0, 1\}^{s_2} \to \{0, 1\}^{\poly(n)}$ be the generator of \cref{lem:nis94} with space bound $\lceil \log \size(\mathcal{P}) \rceil$, so that $s_1 \leq O(S^2)$. 
		
		Let $\ell = \lceil \log_2(1/\delta) \rceil + 2s_1 + 2$, and let $\mathsf{GUVExt}: \{0, 1\}^{\ell} \times \{0, 1\}^d \to \{0, 1\}^{s_1}$ be the $(2s_1, 0.1)$-extractor of \cref{thm:guv09}, so that $d \leq O(\log \log(1/\delta) + \log S)$. On input $x \in \{0, 1\}^n$ and random string $y \in \{0, 1\}^{\ell}$:
		\begin{enumerate}
			\item For every $z \in \{0, 1\}^d$:
			\begin{enumerate}
				\item Let $y_1 = \mathsf{GUVExt}(y, z)$.
				\item Run $\mathsf{Check}(\mathcal{P}, x, y_1)$, where $\mathsf{Check}$ is the algorithm from \cref{lem:nis94}.
				\item If $\mathsf{Check}$ accepts, run $\mathcal{B}(x, \mathsf{NisGen}(y_1, y_2))$ for every $y_2$, take a majority vote, and output the answer.
			\end{enumerate}
			\item Output $\bot$.
		\end{enumerate}
		Clearly, this algorithm runs in space $O(S + d)$. Since $\delta$ is constructible in space $O(S)$, its denominator must have at most $2^{O(S)}$ digits. Therefore, $\delta \geq 2^{-2^{O(S)}}$ and $d \leq O(S)$, so the algorithm runs in space $O(S)$. Furthermore, the algorithm is clearly zero-error. Finally, by \cref{prop:zuc97}, the number of $y$ such that $\mathsf{Check}(\mathcal{P}, x, y_1)$ rejects for every $z$ is at most $2^{2s_1 + 2}$, and hence the failure probability of the algorithm is at most $\frac{2^{2s_1 + 2}}{2^{\ell}} \leq \delta$.
	\end{proof}
	
	We now give our generalization of \cref{thm:fk06-improved}. From the work of Goldreich and Wigderson \cite{gw02}, it follows that if a language $L \in \mathbf{BPSPACE}(S)$ is in $\mathbf{DPSPACE}(S)/a$ for $a \ll n$ via an algorithm where most advice strings are ``good'', then $L$ is close to being in $\mathbf{DPSPACE}(S)$. Our theorem is a \emph{converse}\footnote{The statement of \cref{thm:fk06-improved-generalized} doesn't mention it, but indeed, in the proof of \cref{thm:fk06-improved-generalized}, most advice strings are ``good''.} to this result, showing that in the space-bounded setting, there is a very tight connection between typically-correct derandomizations and simulations with small amounts of advice.
	\begin{theorem} \label{thm:fk06-improved-generalized}
		Fix functions $S: \N \to \N$ with $S(n) \geq \log n$ and $\epsilon: \N \to [0, 1]$ that are constructible in space $O(S)$. Suppose a language $L \in \mathbf{BPSPACE}(S)$ is within $\epsilon$ of $\mathbf{DSPACE}(S)$. Then
		\begin{equation}
			L \in \mathbf{DSPACE}(S)/(n - \log_2(1/\epsilon(n)) + O(S^2)).
		\end{equation}
	\end{theorem}

	\begin{proof}
		Let $\mathcal{A}$ be the algorithm of \cref{lem:zp.space} with $\delta < 2^{-n}/\epsilon$. Let $m = m(n)$ be the number of random bits used by $\mathcal{A}$. Let $\mathcal{B}$ be the algorithm witnessing the fact that $L$ is within $\epsilon$ of $\mathbf{DSPACE}(S)$.
		
		The algorithm with advice is very simple. Given input $x \in \{0, 1\}^n$ and advice $a \in \{0, 1\}^m$, output $\mathcal{A}(x, a)$, unless $\mathcal{A}(x, a) = \bot$, in which case output $\mathcal{B}(x)$. This algorithm clearly runs in $O(S)$ space and uses $n - \log_2(1/\epsilon(n)) + O(S^2)$ bits of advice.
		
		Now we argue that there is some advice string such that the algorithm succeeds on all inputs. Let $S \subseteq \{0, 1\}^n$ be the set of inputs on which $\mathcal{B}$ fails. Consider picking an advice string $a$ uniformly at random. For each string $x \in S$, $\Pr_a[\mathcal{A}(x, a) = \bot] \leq \delta$. Therefore, by the union bound, the probability that there is some $x \in S$ such that $\mathcal{A}(x, a) = \bot$ is at most $|S| \delta = \epsilon \cdot 2^n \cdot \delta < 1$. Therefore, there is \emph{some} advice string such that the algorithm succeeds on all inputs in $S$. Finally, for \emph{any} advice string, the algorithm succeeds on all inputs in $\{0, 1\}^n \setminus S$, because $\mathcal{A}$ is zero-error.
	\end{proof}

	Combining \cref{thm:fk06-improved-generalized} with our typically-correct derandomizations gives unconditional simulations with fewer than $n$ bits of advice:
	\begin{corollary}
		For every constant $c \in \N$,
		\begin{equation}
			\mathbf{BPTISP}(n \polylog n, \log n) \subseteq \mathbf{L}/(n - \log^c n).
		\end{equation}
	\end{corollary}

	\begin{proof}
		Combine \cref{cor:bptisp-zero-coins,thm:fk06-improved-generalized}.
	\end{proof}

	\begin{corollary}
		For every constant $c \in \N$,
		\begin{equation}
			\mathbf{BPTISP}_{\normalfont\text{TM}}(n \polylog n, \log n) \subseteq \mathbf{L}/\left(\frac{n}{\log^c n}\right).
		\end{equation}
	\end{corollary}

	\begin{proof}
		Combine \cref{cor:bptisp-tm-1-zero-coins,thm:fk06-improved-generalized}.
	\end{proof}

	\begin{corollary}
		\begin{equation}
			\mathbf{BPTISP}_{\normalfont\text{TM}}(n^{1.99}, \log n) \subseteq \mathbf{L}/(n - n^{\Omega(1)}).
		\end{equation}
	\end{corollary}

	\begin{proof}
		Combine \cref{cor:bptisp-tm-2,thm:fk06-improved-generalized}.
	\end{proof}
	
	\section{Disambiguating efficient nondeterministic algorithms} \label{sec:nl-ul}
	
	\subsection{Overview}
	Recall that a nondeterministic algorithm is \emph{unambiguous} if on every input, there is at most one accepting computation. Suppose a language $L$ can be decided by a nondeterministic algorithm that runs in time $T = T(n) \geq n$ and space $S = S(n) \geq \log n$. Allender, Reinhardt, and Zhou showed that if $\mathsf{SAT}$ has exponential circuit complexity, there is an unambiguous algorithm for $L$ that runs in space $O(S)$ \cite{arz99}. Unconditionally, van Melkebeek and Prakriya recently gave an unambiguous algorithm for $L$ that runs in time $2^{O(S)}$ and space $O(S \sqrt{\log T})$ \cite{vmp17}.
	
	For some of our results on derandomizing efficient algorithms, we give a corresponding theorem for disambiguating efficient nondeterministic algorithms, albeit with slightly worse parameters.
	
	\subsubsection{Our results}
	
	Let $\mathbf{NTISP}(T, S)$ denote the class of languages that can be decided by a nondeterministic random-access Turing machines that runs in time $T$ and space $S$. Define $\mathbf{UTISP}(T, S)$ the same way, but with the additional requirement that the algorithm is unambiguous. In \cref{sec:disambiguate-bp,sec:ntisp}, we show that for every $S$ and every constant $c \in \N$,
	\begin{equation} \label{eqn:ntisp}
	\mathbf{NTISP}(n \cdot \poly(S), S) \text{ is within } 2^{-S^c} \text{ of } \mathbf{UTISP}(2^{O(S)}, S \sqrt{\log S}).
	\end{equation}
	\Cref{eqn:ntisp} is analogous to \cref{cor:bptisp-zero-coins}.

	
	Reinhardt and Allender showed that $\mathbf{NL} \subseteq \mathbf{UL}/\poly$ \cite{ra00}. In \cref{sec:ul-advice}, we improve the Reinhardt-Allender theorem by showing that $\mathbf{NL} \subseteq \mathbf{UL}/(n + O(\log^2 n))$. More generally, we show that if a language $L \in \mathbf{NSPACE}(S)$ is within $\epsilon(n)$ of being in $\mathbf{USPACE}(S)$, then $L \in \mathbf{USPACE}(S)/(n - \log_2(1/\epsilon(n)) + O(S^2))$. This result is analogous to \cref{thm:fk06-improved-generalized}.
	
	\subsubsection{Techniques}
	
	Our disambiguation theorems are proven using the same ``out of sight, out of mind'' technique that we used in \cref{sec:bptisp,sec:bptisp-tm-1} for derandomization. Roughly, this is possible because of prior work \cite{ra00, vmp17} that reduces the problem of disambiguating algorithms to certain derandomization problems. We review the necessary background in \cref{sec:vmp-overview}.
	
	Our disambiguation algorithms do not really introduce any additional novel techniques, beyond what we already used in \cref{sec:bptisp,sec:bptisp-tm-1}. Rather, our contribution in this section is to identify another setting where our techniques are helpful, thereby illustrating the generality of our techniques.
	
	\subsection{Preliminaries}
	Unambiguous algorithms can be \emph{composed} as long as the inner algorithm is ``single-valued'', which we now define. This notion corresponds to classes such as $\mathbf{UL} \cap \mathbf{coUL}$.
	\begin{definition}
		A \emph{single-valued unambiguous algorithm} $\mathcal{A}$ is a nondeterministic algorithm such that for every input $x$, all but one computation path outputs a special symbol $\botn$ (indicating that the nondeterministic choices were ``bad''). We let $\mathcal{A}(x)$ denote the output of the one remaining computation path.
	\end{definition}
	When describing unambiguous algorithms, we will often include steps such as ``Compute $a = \mathcal{A}(x)$'', where $\mathcal{A}$ is a single-valued unambiguous algorithm. Such a step should be understood as saying to run $\mathcal{A}$ on input $x$. If $\mathcal{A}$ outputs $\botn$, immediately halt and output $\botn$. Otherwise, let $a$ be the output of $\mathcal{A}$.
	
	\subsection{Unambiguous algorithms for connectivity by van Melkebeek and Prakriya} \label{sec:vmp-overview}
	
	Recall that the \emph{s-t connectivity problem} is defined by
	\begin{equation}
	\stconn = \{(G, s, t) : \text{there is a directed path from $s$ to $t$}\},
	\end{equation}
	where $G$ is a digraph and $s, t \in V(G)$. $\stconn$ is a classic example of an $\mathbf{NL}$-complete language \cite{jon75}. Using an ``inductive counting'' technique, Reinhardt and Allender gave a single-valued unambiguous algorithm for testing whether a given digraph is ``min-unique'', as well as a single-valued unambiguous algorithm for solving $\stconn$ in min-unique digraphs \cite{ra00}. Using the isolation lemma, Reinhardt and Allender showed that assigning random weights to a digraph makes it ``min-unique'' \cite{ra00}. These two results are the main ingredients in the proof that $\mathbf{NL} \subseteq \mathbf{UL}/\poly$ \cite{ra00}.
	
	Recently, van Melkebeek and Prakriya gave a ``pseudorandom weight generator'' with seed length $O(\log^2 n)$ \cite{vmp17}.\footnote{In the terminology of van Melkebeek and Prakriya \cite{vmp17}, here we refer to the ``hashing only'' approach.} Just like uniform random weights, the weights produced by this generator make a digraph ``min-unique'' with high probability.\footnote{The van Melkebeek-Prakriya generator only works for \emph{layered} digraphs, but this technicality does not matter for us.}
	
	Roughly, this pseudorandom weight generator by van Melkebeek and Prakriya will play a role in our disambiguation results that is analogous to the role that Nisan's generator played in our derandomization results.
	
	For our purposes, it is not necessary to give a precise account of min-uniqueness. What matters is that $\stconn$ can be decided in unambiguous log-space given two-way access to an $O(\log^2 n)$-bit random string. Furthermore, ``bad'' random strings can be unambiguously detected. We now state this result more carefully.
	
	\begin{theorem}[\cite{vmp17}] \label{thm:vmp-generator}
		There is a single-valued unambiguous algorithm $\mathsf{vMPSeededAlg}$ so that for every $x \in \{0, 1\}^n$,
		\begin{align}
			\Pr_{y \in \{0, 1\}^{\infty}}[\mathsf{vMPSeededAlg}(x, y) \in \{\stconn(x), \botr\}] &= 1, \\
			\Pr_{y \in \{0, 1\}^{\infty}}[\mathsf{vMPSeededAlg}(x, y) = \botr] &\leq 1/2.
		\end{align}
		Furthermore, $\mathsf{vMPSeededAlg}(x, y)$ only reads the first $O(\log^2 n)$ bits of $y$ (the ``seed'') and runs in space $O(\log n)$.
	\end{theorem}

	\begin{proof}[Proof sketch]
		We assume that the reader is familiar with the paper by van Melkebeek and Prakriya \cite{vmp17}. Given an instance $x$ of $\stconn$, the algorithm $\mathsf{vMPSeededAlg}$ first applies a reduction, giving a \emph{layered} digraph $G$ on which to test connectivity. Then, the first $O(\log^2 n)$ bits of $y$ are interpreted as specifying $O(\log n)$ hash functions, which are used to assign weights to the vertices in $G$. An algorithm by Reinhardt and Allender \cite{ra00} is run to determine whether the resulting weighted digraph is min-unique. If it is not, $\mathsf{vMPSeededAlg}$ outputs $\botr$. If it is, another closely related algorithm by Reinhardt and Allender \cite{ra00} is run to decide connectivity in the resulting weighted digraph.
	\end{proof}

	Notice that $\mathsf{vMPSeededAlg}$ can be thought of as having \emph{three} read-only inputs: the ``real'' input $x \in \{0, 1\}^n$; the random seed $y \in \{0, 1\}^{O(\log^2 n)}$; and the nondeterministic bits $z \in \{0, 1\}^{\poly(n)}$. The algorithm has two-way access to $x$ and $y$ and one-way access to $z$. Notice also that a computation path of $\mathsf{vMPSeededAlg}$ has \emph{four} possible outputs: $0$, indicating that $x \not \in \stconn$; $1$, indicating that $x \in \stconn$; $\botn$, indicating bad nondeterministic bits $z$; and $\botr$, indicating bad random bits $y$.
	
	Iterating over all $y$ in \cref{thm:vmp-generator} would take $\Theta(\log^2 n)$ space. By modifying their ``pseudorandom weight generator'', van Melkebeek and Prakriya gave an unambiguous algorithm for $\stconn$ that runs in $O(\log^{3/2} n)$ space. The performance of their algorithm is improved if we only need to search for \emph{short} paths; the precise details are given by the following theorem.
	
	\begin{theorem}[\cite{vmp17}] \label{thm:vmp-alg-2}
		There is a single-valued unambiguous algorithm $\mathsf{vMPShortPathsAlg}$ such that if $G$ is a digraph, $s, t \in V(G)$, and $r \in \N$, then $\mathsf{vMPShortPathsAlg}(G, s, t, r) = 1$ if and only if there is a directed path from $s$ to $t$ in $G$ of length at most $r$. Furthermore, $\mathsf{vMPShortPathsAlg}$ runs in time $\poly(n)$ and space $O(\log n \sqrt{\log r})$.
	\end{theorem}

	\begin{proof}[Proof sketch]
		Again, we assume that the reader is familiar with the paper by van Melkebeek and Prakriya \cite{vmp17}. Again, we first apply a reduction, giving a layered digraph $G'$ of width $|V(G)|$ and length $r$, so that the question is whether there is a path from the first vertex in the first layer to the first vertex in the last layer.
		
		We rely on the ``combined hashing and shifting'' generator by van Melkebeek and Prakriya \cite[Theorem 1]{vmp17}. The seed of this generator specifies $O(\sqrt{\log r})$ hash functions (each is specified with $O(\log n)$ bits). We find these hash functions by exhaustive search one at a time, maintaining the invariant that portions of $G'$ that have weights assigned are min-unique. We test for min-uniqueness using a slight variant of the algorithm by Reinhardt and Allender \cite{ra00} described by van Melkebeek and Prakriya \cite[Lemma 1]{vmp17}.
	\end{proof}

	Roughly speaking, \cref{thm:vmp-alg-2} plays a role in our disambiguation results that is analogous to the role that the Nisan-Zuckerman generator played in our derandomization results.
	
	
	\subsection{Disambiguating branching programs} \label{sec:disambiguate-bp}
	
	For us, a \emph{nondeterministic branching program} $\mathcal{P}$ on $\{0, 1\}^n \times \{0, 1\}^m$ is a randomized branching program (but we think of the second input to the program as nondeterministic bits instead of random bits) such that some vertex $v_0 \in V(\mathcal{P})$ is labeled as the \emph{start vertex} and some vertex $v_{\text{accept}} \in V(\mathcal{P})$ is labeled as the \emph{accepting vertex}. We identify $\mathcal{P}$ with a function $\mathcal{P}: \{0, 1\}^n \times \{0, 1\}^m \to \{0, 1\}$ defined by
	\begin{equation}
		\mathcal{P}(x, y) = \begin{cases}
		1 & \text{if } \mathcal{P}(v_0; x, y) = v_{\text{accept}}, \\
		0 & \text{otherwise,}
		\end{cases}
	\end{equation}
	and we \emph{also} identify $\mathcal{P}$ with a function $\mathcal{P}: \{0, 1\}^n \to \{0, 1\}$ defined by
	\begin{equation} \label{eqn:nondeterminism}
		\mathcal{P}(x) = 1 \iff \exists y\; \mathcal{P}(x, y) = 1.
	\end{equation}
	(\Cref{eqn:nondeterminism} expresses the fact that $\mathcal{P}$ is a \emph{nondeterministic} branching program.)	Finally, an \emph{R-OW} nondeterministic branching program is just a nondeterministic branching program that is R-OW when thought of as a randomized branching program, i.e., it reads its nondeterministic bits from left to right.
	
	\begin{theorem} \label{thm:simulate-nondeterministic-branching-program}
		For every constant $c \in \N$, there is a single-valued unambiguous algorithm $\mathsf{A}$ with the following properties. Suppose $\mathcal{P}$ is an R-OW nondeterministic branching program on $\{0, 1\}^n \times \{0, 1\}^T$. Suppose $S \geq \log n$, where $S \stackrel{\normalfont\text{def}}{=} \lceil \log \size(\mathcal{P}) \rceil$, and $T \geq \length(\mathcal{P})$. Then
		\begin{equation}
			\density\{x \in \{0, 1\}^n : \mathsf{A}(\mathcal{P}, x, T) \neq \mathcal{P}(x)\} \leq 2^{-S^c}.
		\end{equation}
		Furthermore, $\mathsf{A}(\mathcal{P}, x, T)$ runs in time $2^{O(S)}$ and space $O(S \sqrt{\log \lceil T/n \rceil + \log S})$.
	\end{theorem}
	
	Toward proving \cref{thm:simulate-nondeterministic-branching-program}, we introduce some notation. The computation of $\mathcal{P}(x)$ naturally reduces to $\stconn$. Let $\mathcal{P}[x]$ be the digraph $(V, E)$, where $V = V(\mathcal{P})$ and $E$ is the set of edges $(u, v)$ in $\mathcal{P}$ labeled with $x_{i(u)}0$ or $x_{i(u)}1$. (So every nonterminal vertex in $\mathcal{P}[x]$ has outdegree $2$.) That way, $\mathcal{P}(x) = 1$ if and only if $(\mathcal{P}[x], v_0, v_{\text{accept}}) \in \stconn$.
	
	\begin{figure}
		\begin{framed}
			\begin{enumerate}
				\item If $S^{c + 1} > n$, output $\mathsf{vMPShortPathsAlg}(\mathcal{P}[x], v_0, v_{\text{accept}}, T)$. Otherwise:
				\item Let $I_1, I_2, \dots, I_B \subseteq [n]$ be disjoint sets of size $S^{c + 1}$ with $B$ as large as possible.
				\item For $b = 1$ to $B$:
				\begin{enumerate}
					\item Let $I = I_b$. Let $V_b = \{v \in V(\mathcal{P}) : i(v) \in I\} \cup \{v_0, v_{\text{accept}}\}$. Let $E_b$ be the set of pairs $(u, v) \in V_b^2$ such that there is a directed path from $u$ to $v$ in $\mathcal{P} \vert_{[n] \setminus I}[x]$. Let $H_b$ be the digraph $(V_b, E_b)$.
					\item Compute $a \stackrel{\text{def}}{=} \mathsf{vMPShortPathsAlg}(H_b, v_0, v_{\text{accept}}, \lfloor S^{c + 1} T/n \rfloor + 1)$. Whenever $\mathsf{vMPShortPathsAlg}$ asks whether some pair $(u, v)$ is in $E_b$, run $\mathsf{B}(\mathcal{P}, x, b, u, v)$, where $\mathsf{B}$ is the algorithm of \cref{fig:eb-membership}.
					\item If $a = 1$, halt and output $1$.
				\end{enumerate}
				\item Output $0$.
			\end{enumerate}
			\vspace{-5mm}
		\end{framed}
		\vspace{-5mm}
		\caption{The algorithm $\mathsf{A}$ of \cref{thm:simulate-nondeterministic-branching-program}.} \label{fig:simulate-nondeterministic-branching-program}
	\end{figure}

	\begin{figure}
		\begin{framed}
			\begin{enumerate}
				\item For every $y \in \{0, 1\}^{O(\log S)}$:
				\begin{enumerate}
					\item Let $a' = \mathsf{vMPSeededAlg}(\mathcal{P} \vert_{[n] \setminus I}[x], u, v, \mathsf{GUVExt}(x \vert_I, y))$.
					\item If $a' \neq \botr$, halt and output $a'$.
				\end{enumerate}
				\item Output $\boti$.
			\end{enumerate}
			\vspace{-5mm}
		\end{framed}
		\vspace{-5mm}
		\caption{The algorithm $\mathsf{B}$ used by $\mathsf{A}$ to decide whether $(u, v) \in E_b$. The block $I$ is the same block $I_b$ used by $\mathsf{A}$.} \label{fig:eb-membership}
	\end{figure}

	The algorithm $\mathsf{A}$ of \cref{thm:simulate-nondeterministic-branching-program} is given in \cref{fig:simulate-nondeterministic-branching-program}. The algorithm relies on a subroutine $\mathsf{B}$ given in \cref{fig:eb-membership}.
	
	\paragraph{Parameters}
	Let $s$ be the number of random bits used by $\mathsf{vMPSeededAlg}$, so that $s \leq O(S^2)$. The subroutine $\mathsf{B}$ relies on the extractor $\mathsf{GUVExt}$ of \cref{thm:guv09}. This extractor is instantiated with source length $\ell \stackrel{\text{def}}{=} S^{c + 1}$, error $0.1$, entropy $k \stackrel{\text{def}}{=} 2s$, and output length $s$. The seed length of $\mathsf{GUVExt}$ is $d \leq O(\log \ell) = O(\log S)$.
	
	\paragraph{Efficiency}
	First, we bound the space complexity of $\mathsf{A}$. If $S^{c + 1} > n$, then $\mathsf{A}$ runs in space
	\begin{equation}
		O(\log \size(\mathcal{P}) \sqrt{\log T}) = O(S \sqrt{\log T}) \leq O\left(S \sqrt{\log \frac{TS^{c + 1}}{n}}\right) = O(S \sqrt{\log(T/n) + \log S}).
	\end{equation}
	Suppose now that $S^{c + 1} \leq n$. The extractor $\mathsf{GUVExt}$ runs in space $O(\log S)$, and $\mathsf{vMPSeededAlg}$ runs in space $O(S)$, so $\mathsf{B}$ runs in space $O(S)$. The algorithm $\mathsf{vMPShortPathsAlg}$ runs in space
	\begin{align}
		O(\log |V_b| \sqrt{\log(\lfloor S^{c + 1}T/n \rfloor + 1)}) &\leq O(S \sqrt{\log \lceil T/n \rceil + \log S}).
	\end{align}
	Therefore, overall, $\mathsf{A}$ runs in space $O(S \sqrt{\log \lceil T/n \rceil + \log S})$.
	
	Next, we bound the running time of $\mathsf{A}$. If $S^{c + 1} > n$, then $\mathsf{A}$ runs in time $\poly(\size(\mathcal{P})) = 2^{O(S)}$ as claimed. Suppose now that $S^{c + 1} \leq n$. Because $\mathsf{B}$ runs in space $O(S)$, it must run in time $2^{O(S)}$. Therefore, $\mathsf{vMPShortPathsAlg}$ runs in time $2^{O(S)} \cdot 2^{O(S)} = 2^{O(S)}$. Therefore, overall, $\mathsf{A}$ runs in time $2^{O(S)}$.
	
	\paragraph{Correctness}
	
	Since $\mathsf{vMPSeededAlg}$ and $\mathsf{vMPShortPathsAlg}$ are single-valued unambiguous algorithms, $\mathsf{A}$ is a single-valued unambiguous algorithm. All that remains is to show that for most $x$, $\mathsf{A}(\mathcal{P}, x, T) = \mathcal{P}(x)$. First, we show that for most $x$, the subroutine $\mathsf{B}$ is correct, i.e., the one computation path that does not output $\botn$ outputs a bit indicating whether $(u, v) \in E_b$. Clearly, the only way that $\mathsf{B}$ can be incorrect is if it outputs $\boti$, indicating a ``hard'' input $x$.

	\begin{claim} \label{clm:subroutine-correctness}
		For every $\mathcal{P}$,
		\begin{equation}
			\density\{x \in \{0, 1\}^n : \exists b, u, v \text{ such that } \mathsf{B}(\mathcal{P}, x, b, u, v) = \boti\} \leq 2^{-S^c}.
		\end{equation}
	\end{claim}

	\begin{proof}
		The graph $\mathcal{P}\vert_{[n] \setminus I}[x]$ does not depend on $x \vert_I$. Therefore, for each fixed $b$, each fixed $z \in \{0, 1\}^{n - |I_b|}$, and each fixed $u, v \in V(\mathcal{P})$, by \cref{prop:zuc97},
		\begin{equation}
			\#\{x : x \vert_{[n] \setminus I} = z \text{ and } \mathsf{B}(\mathcal{P}, x, b, u, v) = \boti\} \leq 2^{k + 2} \leq 2^{O(S^4)}.
		\end{equation}
		Therefore, by summing over all $b, z, u, v$,
		\begin{align}
			\#\{x \in \{0, 1\}^n : \exists b, u, v \text{ such that } \mathsf{B}(\mathcal{P}, x, b, u, v) = \boti\} &\leq 2^{n - S^{c + 1} + \log n + 2S + O(S^4)} \\
			&= 2^{n - S^{c + 1} + O(S^4)} \\
			&\leq 2^{n - S^c}
		\end{align}
		for sufficiently large $n$.
	\end{proof}

	Next, we show that as long as $\mathsf{B}$ does not make any mistakes, $\mathsf{A}$ is correct.

	\begin{claim} \label{clm:outer-routine-correctness}
		If $\mathcal{P}(x) = 1$, there is some $b \in [B]$ so that there is a path from $v_0$ to $v_{\text{accept}}$ through $H_b$ of length at most $\lfloor S^{c + 1} T/n \rfloor + 1$.
	\end{claim}

	\begin{proof}
		Since $\mathcal{P}(x) = 1$, there is a path from $v_0$ to $v_{\text{accept}}$ through $\mathcal{P}[x]$ of length at most $T$. Let $v_0, v_1, v_2, \dots, v_{T'} = v_{\text{accept}}$ be the vertices visited by that path, so that $T' \leq T$. Consider picking $b \in [B]$ uniformly at random. Then for each $t < T'$, $\Pr[i(v_t) \in I_b] \leq S^{c + 1}/n$. Therefore, by linearity of expectation,
		\begin{equation}
			\E[\#\{t : i(v_t) \in I_b\}] \leq S^{c + 1} T/n.
		\end{equation}
		The best case is at least as good as the average case, so there is some $b \in [B]$ such that $\#\{t : i(v_t) \in I_b\} \leq S^{c + 1} T/n$. Let $t_1, t_2, \dots, t_r$ be the indices $t$ such that $i(v_t) \in I_b$. Then by the definition of $E_b$, the edges $(v_0, t_1), (t_1, t_2), \dots, (t_{r - 1}, t_{r}), (t_{r}, v_{\text{accept}})$ are all present in $H_b$. Therefore, there is a path from $v_0$ to $v_{\text{accept}}$ through $H_b$ of length at most $r + 1$.
	\end{proof}

	Combining \cref{clm:subroutine-correctness,clm:outer-routine-correctness} completes the proof of \cref{thm:simulate-nondeterministic-branching-program}.
	
	\subsection{Disambiguating uniform random-access algorithms} \label{sec:ntisp}
	
	\begin{corollary} \label{cor:ntisp}
		For every space-constructible function $S(n) \geq \log n$, for every constant $c \in \N$,
		\begin{equation}
			\mathbf{NTISP}(n \cdot \poly(S), S) \text{ is within } 2^{-S^c} \text{ of } \mathbf{UTISP}(2^{O(S)}, S \sqrt{\log S}).
		\end{equation}
	\end{corollary}

	\begin{proof}[Proof sketch]
		The class $\mathbf{NTISP}(n \cdot \poly(S), S)$ corresponds to R-OW nondeterministic branching programs of size $2^{O(S)}$ and length $T = n \cdot \poly(S)$. For these parameters, the algorithm of \cref{thm:simulate-nondeterministic-branching-program} runs in time $2^{O(S)}$ and space $O(S \sqrt{\log S})$.
	\end{proof}

	\subsection{Disambiguation with advice} \label{sec:ul-advice}
	
	We now show how to disambiguate $\mathbf{NL}$ with only $n + O(\log^2 n)$ bits of advice. The proof is very similar to the proof of \cref{thm:fk06-improved}.
	
	\begin{theorem} \label{thm:ra00-improved}
		$\mathbf{NL} \subseteq \mathbf{UL}/(n + O(\log^2 n))$.
	\end{theorem}

	\begin{proof}
		Let $\mathcal{R}$ be a log-space reduction from $L \in \mathbf{NL}$ to $\stconn$. Let $s$ be the number of random bits used by $\mathsf{vMPSeededAlg}$ on inputs of length $n^c$, where $n^c$ is the length of outputs of $\mathcal{R}$ on inputs of length $n$. Let $\mathsf{GUVExt}: \{0, 1\}^{n + 2s + 3} \times \{0, 1\}^d \to \{0, 1\}^s$ be the $(2s, 0.1)$-extractor of \cref{thm:guv09}, so that $d \leq O(\log n)$.
		
		Given input $x \in \{0, 1\}^n$ and advice $a \in \{0, 1\}^{n + 2s + 3}$, compute
		\begin{equation}
			a_z \stackrel{\text{def}}{=} \mathsf{vMPSeededAlg}(\mathcal{R}(x), \mathsf{GUVExt}(a, z))
		\end{equation}
		for all $z$ and accept if there is some $z$ so that $a_z = 1$.
		
		This algorithm clearly runs in space $O(\log n)$ and is unambiguous. By \cref{prop:zuc97}, for each fixed $x$, the number of advice strings $a$ causing the algorithm to give the wrong answer is at most $2^{2s + 2}$. Therefore, the total number of advice strings $a$ that cause the algorithm to give the wrong answer for any $x$ is at most $2^{n + 2s + 2} < 2^{|a|}$. Therefore, there is some choice of $a$ such that the algorithm succeeds on all inputs.
	\end{proof}
	
	Just like we did with \cref{thm:fk06-improved}, we now generalize \cref{thm:ra00-improved}, showing that the amount of advice can be reduced to below $n$ if we start with a language that has a typically-correct disambiguation.
	
	\begin{theorem} \label{thm:ra00-improved-generalized}
		Fix functions $S: \N \to \N$ with $S(n) \geq \log n$ and $\epsilon: \N \to [0, 1]$ that are constructible in $O(S)$ space. Suppose a language $L \in \mathbf{NSPACE}(S)$ is within $\epsilon$ of $\mathbf{USPACE}(S)$. Then
		\begin{equation}
		L \in \mathbf{USPACE}(S)/(n - \log_2(1/\epsilon(n)) + O(S^2)).
		\end{equation}
	\end{theorem}

	The proof of \cref{thm:ra00-improved-generalized} is very similar to the proof of \cref{thm:fk06-improved-generalized}. Because the proof of \cref{thm:ra00-improved-generalized} does not introduce any significantly new techniques, we defer the proof to \cref{apx:ra00-improved-generalized}.

	\begin{corollary}
		For every constant $c \in \N$,
		\begin{equation}
			\mathbf{NTISP}(n \polylog n, \log n) \subseteq \mathbf{USPACE}(\log n \sqrt{\log \log n})/(n - \log^c n).
		\end{equation}
	\end{corollary}

	\begin{proof}
		For any $L \in \mathbf{NTISP}(n \polylog n, \log n)$, obviously $L \in \mathbf{NSPACE}(\log n \sqrt{\log \log n})$, and by \cref{cor:ntisp}, $L$ is within $2^{-\log^c n}$ of $\mathbf{USPACE}(\log n \sqrt{\log \log n})$. Applying \cref{thm:ra00-improved-generalized} completes the proof.
	\end{proof}


	
	\section{Directions for further research} \label{sec:future}
	The main open problem in this area is to prove that $\mathbf{BPL}$ is within $o(1)$ of $\mathbf{L}$. \cref{cor:bptisp-zero-coins} implies that $\mathbf{BPTISP}(n \polylog n, \log n)$ is within $o(1)$ of $\mathbf{L}$, and \cref{cor:bptisp-tm-2} implies that $\mathbf{BPTISP}_{\text{TM}}(n^{1.99}, \log n)$ is within $o(1)$ of $\mathbf{L}$, but $\mathbf{BPL}$ allows time $n^c$ where $c$ is an arbitrarily large constant. At present, for a generic language $L \in \mathbf{BPL}$, we do not even know a deterministic log-space algorithm that succeeds on at least \emph{one} input of each length.
	
	
	This work also provides some additional motivation for studying small-space extractors. The two extractors we used in this paper (\cref{thm:su05,thm:guv09}) were sufficient for our applications, but it would be nice to have a single log-space extractor that is optimal up to constants for the full range of parameters.
	
	\section{Acknowledgments}
	
	We thank Michael Forbes, Scott Aaronson, David Zuckerman, Adam Klivans, and Anna G{\'a}l for helpful comments on an early draft of this paper. We thank Amnon Ta-Shma, Lijie Chen, Chris Umans, David Zuckerman, Adam Klivans, Anna G{\'a}l, Gil Cohen, Shachar Lovett, Oded Goldreich, and Avi Wigderson for helpful discussions.
	
	\bibliographystyle{alpha}
	\bibliography{bptisp}
	
	\appendix
	\section{Proof of \cref{thm:su05}: The Shaltiel-Umans extractor} \label{apx:su05}
	
	In this section, we discuss the proof of \cref{thm:su05}. The extractor follows the same basic construction that Shaltiel and Umans used for a ``low error'' extractor \cite[Corollary 4.21]{su05}. We will assume that the reader is familiar with the paper by Shaltiel and Umans \cite{su05}. We will also switch to the parameter names by Shaltiel and Umans, so the source length of the extractor is $n$ rather than $\ell$, and the seed length is $t$ rather than $d$. In these terms, we are shooting for time $\poly(n)$ and space $O(t)$.
	
	The only change to the \emph{construction} that we make is that we will use a different instantiation of the ``base field'' $\F_q$. Shaltiel and Umans \cite{su05} used a deterministic algorithm by Shoup that finds an irreducible polynomial of degree $\log q$ over $\F_2$ in time $\poly(\log q)$. Unfortunately, Shoup's algorithm is not sufficiently space-efficient for our purposes. To get around this issue, we use an extremely explicit family of irreducible polynomials:
	\begin{lemma}[{\cite[Theorem 1.1.28]{vl99}}] \label{lem:vl99}
		For every $a \in \N$, the polynomial $x^{2 \cdot 3^a} + x^{3^a} + 1$ is irreducible over $\F_2$.
	\end{lemma}
	Therefore, by replacing $q$ by some power of two between $q$ and $q^3$, we can easily, deterministically construct an irreducible polynomial of degree $\log q$ in time $\poly(\log q)$ and space $O(\log q)$. This only affects the bit length of field elements, $\log q$, by at most a factor of $3$. Therefore, the hypotheses of Shaltiel and Umans' main technical theorem \cite[Theorem 4.5]{su05} are still met, so the extractor is still correct.
	
	Now we turn to analyzing the efficiency of the extractor. The parameters $h, d, m, \rho, q$ used by Shaltiel and Umans (with the described modification to $q$) can all easily be computed in time $\poly(n)$ and space $O(t)$. Next, we inspect the construction of the matrix $B$ used by Shaltiel and Umans \cite[Proof of Lemma 4.18]{su05}. The exhaustive search used to find the irreducible polynomial $p(z)$ takes space $O(d \log q) \leq O(t)$. The exhaustive search used to find the generator $g$ for $(H^d)^{\times}$ also takes space $O(d \log q) = O(t)$. Finally, multiplication by $g$ takes space $O(d \log q) = O(t)$.
	
	It follows immediately that the ``$q$-ary extractor'' $E'$ given by Shaltiel and Umans \cite[Equation 8]{su05} runs in space $O(t)$, because we only need to store the vector $B^i \vec{v}$. Finally, to get from $E'$ to the final extractor, a simple Hadamard code is applied, which can trivially be computed in time $\poly(n)$ and space $O(t)$.

	\section{Proof of \cref{thm:guv09}: The GUV extractor} \label{apx:guv09}
	
	In this section, we discuss the proof of \cref{thm:guv09}. We will assume that the reader is familiar with the paper by Guruswami, Umans, and Vadhan. Recall that a \emph{condenser} is like an extractor, except that the output is merely guaranteed to be close to having high entropy instead of being guaranteed to be close to uniform.
	
	\begin{definition}
		A function $\mathsf{Con}: \{0, 1\}^n \times \{0, 1\}^d \to \{0, 1\}^{n'}$ is a \emph{$k \to_{\epsilon} k'$ condenser} if for every random variable $X$ with $H_{\infty}(X) \geq k$, there exists a distribution $Z$ with $H_{\infty}(Z) \geq k'$ such that if we let $Y \sim U_d$ be independent of $X$, then $\mathsf{Con}(X, Y) \sim_{\epsilon} Z$.
	\end{definition}	

	Guruswami, Umans, and Vadhan constructed a lossy condenser based on folded Reed-Solomon codes~\cite[Theorem 6.2]{guv09}. To ensure space efficiency, we will slightly modify their construction to get the following condenser. We will follow the parameter names by Guruswami, Umans, and Vadhan.
	
	\begin{theorem}[Based on {\cite[Theorem 6.2]{guv09}}] \label{thm:guv-condenser}
		Let $\alpha > 0$ be a constant. Consider any $n \in \N, \ell \leq n$ such that $2^{\ell}$ is an integer and any $\epsilon > 0$. There is a parameter $t = \Theta(\log(n \ell / \epsilon))$ and a  $$ (1 + 1/\alpha) \ell t + \log(1/\epsilon) \to_{3\epsilon} \ell t + d - 2$$ condenser $\mathsf{GUVCon}: \{0, 1\}^n \times \{0, 1\}^d \to \{0, 1\}^{n'}$, computable in space $O(d)$, with seed length $d \leq (1 + 1/\alpha) t$ and output length $n' \leq (1 + 1/\alpha) \ell t + d$, provided $\ell t \geq \log(1/\epsilon)$.
	\end{theorem}

	\begin{proof}[Proof sketch]
		We need to use a base field $\F_q$ based on \cref{lem:vl99}, so we slightly modify the parameters of the GUV construction as follows. Choose $q$ to be the smallest power of two of the form $2^{2 \cdot 3^a}$ such that $q \geq (2^{2 + 1/\alpha} \cdot n\ell / \epsilon)^{1 + \alpha}$. This $q$ satisfies $q \leq (2^{2 + 1/\alpha} \cdot n\ell / \epsilon)^{3 + 3\alpha}$. Next, define $t = \lceil \frac{\alpha \log q}{1 + \alpha} \rceil$ and $h = 2^t$, so that $q \in ((h/2)^{1 + 1/\alpha}, h^{1 + 1/\alpha}]$. Therefore, we still have
		\begin{align}
			q &> h \cdot h^{1/\alpha} / 2^{1 + 1/\alpha} \\
			&\geq h \cdot q^{1 / (1 + \alpha)} / 2^{1 + 1/\alpha} \\
			&\geq 2 h n \ell / \epsilon,
		\end{align}
		and hence $A \geq \epsilon q / 2$. The rest of the argument is as in the original paper \cite{guv09}.
	\end{proof}

	There is a standard extractor based on expander walks that works well for constant error and constant entropy rate. Using the Gutfreund-Viola expander walk (\cref{thm:gv04}), this extractor runs in logarithmic space:
	\begin{lemma} \label{lem:gv04-extractor}
		Let $\alpha, \epsilon > 0$ be constants. There is some constant $\beta \in (0, 1)$ so that for all $n$, there is a $(\beta n, \epsilon)$-extractor $\mathsf{GVExt}: \{0, 1\}^n \times \{0, 1\}^d \to \{0, 1\}^m$ with $t \leq \log(\alpha n)$ and $m \geq (1 - \alpha) n$ so that given $x$ and $y$, $\mathsf{GVExt}(x, y)$ can be computed in $O(\log n)$ space.
	\end{lemma}

	\begin{proof}[Proof sketch]
		This construction of an extractor from an expander is standard; see, e.g., an exposition by Guruswami et al. \cite[Theorem 4.6]{guv09}. The space bound follows from \cref{thm:gv04}.
	\end{proof}

	Finally, \cref{thm:guv09} follows by composing \cref{thm:guv-condenser} and \cref{lem:gv04-extractor}, just as is explained in the paper by Guruswami et al.\ \cite[Theorem 4.7]{guv09}.
	
	\section{Proof of \cref{prop:zuc97}: Extractors are good samplers} \label{apx:zuc97}
	
	Let $X \subseteq \{0, 1\}^{\ell}$ be the set on the left-hand side of \cref{eqn:zuc97}. Since total variation distance is half $\ell_1$ distance, for each $x \in X$,
	\begin{equation}
		\sum_{v \in V} |\Pr[f(U_s) = v] - \Pr[f(\mathsf{Ext}(x, U_d)) = v]| > \epsilon |V|.
	\end{equation}
	Therefore, by the triangle inequality, for each $x \in X$, there is some $v_x \in V$ such that
	\begin{equation} \label{eqn:vx}
		|\Pr[f(U_s) = v_x] - \Pr[f(\mathsf{Ext}(x, U_d)) = v_x]| > \epsilon.
	\end{equation}
	Partition $X = X_1 \cup \dots \cup X_{|V|}$, where $\mathcal{X}_v = \{x \in X : v_x = v\}$. For each $v$, we can further partition $X_v$ into $X_v^+ \cup X_v^-$, based on which term of the left hand side of \cref{eqn:vx} is bigger.
	
	Identify $X_v^+$ with a random variable that is uniformly distributed over the set $X_v^+$, and let $Y \sim U_d$ be independent of $X_v^+$. Then 
	\begin{equation}
		\Pr[\mathsf{Ext}(X_v^+, Y) \in f^{-1}(v_x)] > \Pr[U_s \in f^{-1}(v_x)] + \epsilon.
	\end{equation}
	Therefore, by the extractor condition, $|X_v^+| \leq 2^k$. Similarly, $|X_v^-| \leq 2^k$, and hence $|X_v| \leq 2^{k + 1}$. By summing over all $v$, we conclude that $|X| \leq 2^{k + 1} |V|$ as claimed.
	
	\section{Proof of \cref{thm:simulate-s-ow-branching-program}: Derandomizing S-OW branching programs} \label{apx:simulate-s-ow-branching-program}
	
	The algorithm $\mathsf{A}$ of \cref{thm:simulate-s-ow-branching-program} is given in \cref{fig:simulate-s-ow-branching-program}. The analysis is similar to the proof of \cref{thm:simulate-branching-program}. The main difference is when we argue that the second hybrid distribution, $\mathsf{H}_2$, simulates $\mathcal{P}$. (This argument has just two hybrid distributions.) Details follow.
	
	\begin{figure}
		\begin{framed}
			\begin{enumerate}
				\item If $S^{c + 1} > \sqrt{n}$, directly simulate $\mathcal{P}(v_0; x, U_T)$ using $T$ random bits. Otherwise:
				\item Partition $[n]$ into disjoint blocks, $[n] = I_1 \cup I_2 \cup \dots \cup I_B$, where $|I_b| \approx h$. More precisely, let $B = \lceil n/h \rceil$, and let $I_b = \{h \cdot (b - 1) + 1, h \cdot (b - 1) + 2, \dots, \min\{h \cdot b, n\}\}$. Let $I_0 = I_{B + 1} = \emptyset$.
				\item For $b \in [B]$, let $I'_b = [n] \setminus (I_{b - 1} \cup I_b \cup I_{b + 1})$, with the largest elements removed so that $|I'_b| = n - 3h$.
				\item Initialize $v = v_0$. Repeat $r$ times, where $r \stackrel{\text{def}}{=} \lceil T/h \rceil$:
				\begin{enumerate}
					\item Let $b \in [B]$ be such that $i(v) \in I_b$. Let $I = I'_b$.
					\item Pick $y \in \{0, 1\}^{O(S)}$ uniformly at random.
					\item Let $v = \mathcal{P}\vert_{[n] \setminus I}(v; x, \mathsf{NisGen}(\mathsf{SUExt}(x\vert_I, y)))$.
				\end{enumerate}
				\item Output $v$.
			\end{enumerate}
			\vspace{-5mm}
		\end{framed}
		\vspace{-5mm}
		\caption{The algorithm $\mathsf{A}$ of \cref{thm:simulate-s-ow-branching-program}.} \label{fig:simulate-s-ow-branching-program}
	\end{figure}
	
	\paragraph{Parameters} Just like in the proof of \cref{thm:simulate-branching-program}, we can assume without loss of generality that $T \leq 2^S$. The block size $h$ in \cref{fig:simulate-s-ow-branching-program} is
	\begin{equation}
	h \stackrel{\text{def}}{=} \left\lfloor \frac{n}{3S^{c + 1}} \right\rfloor.
	\end{equation}
	Note that this time, the number of phases, $r$, is $\lceil T/h \rceil$, where $h$ is the \emph{block size}, in contrast to the proof of \cref{thm:simulate-branching-program}, where the number of phases was roughly $T/B$, where $B$ is the \emph{number of blocks}.
	
	The algorithm $\mathsf{A}$ relies on Nisan's generator $\mathsf{NisGen}$ (\cref{thm:nis92}). Naturally, the generator is instantiated with parameters $S, T$ from the statement of \cref{thm:simulate-s-ow-branching-program}. The error of $\mathsf{NisGen}$ is set at $\epsilon \stackrel{\text{def}}{=} \frac{\exp(-cS)}{2r}$, just like in the proof of \cref{thm:simulate-branching-program}. Again, the seed length of $\mathsf{NisGen}$ is $s \leq O(S \log T) \leq O(S^2)$.
	
	The algorithm $\mathsf{A}$ also relies on the Shaltiel-Umans extractor $\mathsf{SUExt}$ of \cref{thm:su05}. This extractor is instantiated with source length $\ell \stackrel{\text{def}}{=} n - 3h$, $\alpha \stackrel{\text{def}}{=} 1/2$, error
	\begin{equation}
	\epsilon' \stackrel{\text{def}}{=} \frac{\exp(-cS)}{r \cdot 2^S},
	\end{equation}
	and entropy $k \stackrel{\text{def}}{=} \sqrt{n}$. This choice of $k$ meets the hypotheses of \cref{thm:su05}, because $\log^{4/\alpha} \ell \leq \log^8 n \leq k$, and $S^{c + 1} \leq \sqrt{n}$, so $\log^{4/\alpha}(1/\epsilon) \leq \polylog n \leq k$. Furthermore, by construction, $k^{1 - \alpha} = n^{1/4} \geq s$ as long as $c \geq 4$ and $n$ is sufficiently large, so we can think of $\mathsf{SUExt}_2$ as outputting $s$ bits.
	
	\paragraph{Efficiency} The runtime analysis of $\mathsf{A}$ is essentially the same as in the proof of \cref{thm:simulate-branching-program}; the only substantial difference is that the input to $\mathsf{SUExt}$ has length $\Theta(n)$, so $\mathsf{SUExt}$ takes $\poly(n)$ time instead of $\poly(S)$ time. Thus, overall, $\mathsf{A}$ runs in time $T \cdot \poly(n, S)$. The space complexity and randomness complexity analyses are essentially the same as in the proof of \cref{thm:simulate-branching-program}.
	
	\paragraph{Correctness} The proof of \cref{eqn:simulate-s-ow-branching-program} has the same structure as the proof of \cref{eqn:simulate-branching-program-correctness}. Assume without loss of generality that $S^{c + 1} \leq \sqrt{n}$. The first hybrid distribution is defined by the algorithm given in \cref{fig:s-ow-hybrid-1}. The number of ``bad'' inputs in \cref{clm:s-ow-a-hybrid-1} is much lower than the number of ``bad'' inputs in \cref{clm:a-hybrid-1}; intuitively, this is because $\mathsf{A}$ uses a \emph{much larger} portion of the input as a source of randomness compared to the algorithm of \cref{thm:simulate-branching-program}.
	
	\begin{figure}
		\begin{framed}
			\begin{enumerate}
				\item Initialize $v = v_0$. Repeat $r$ times, where $r \stackrel{\text{def}}{=} \lceil T/h \rceil$:
				\begin{enumerate}
					\item Let $b \in [B]$ be such that $i(v) \in I_b$. Let $I = I'_b$.
					\item Pick $y' \in \{0, 1\}^s$ uniformly at random.
					\item Let $v = \mathcal{P} \vert_{[n] \setminus I}(v; x, \mathsf{NisGen}(y'))$.
				\end{enumerate}
				\item Output $v$.
			\end{enumerate}
			\vspace{-5mm}
		\end{framed}
		\vspace{-5mm}
		\caption{The algorithm $\mathsf{H}_1$ defining the first hybrid distribution used to prove \cref{eqn:simulate-s-ow-branching-program}. The only difference between $\mathsf{A}$ and $\mathsf{H}_1$ is that $\mathsf{H}_1$ picks a uniform random seed for $\mathsf{NisGen}$, instead of extracting the seed from the input.} \label{fig:s-ow-hybrid-1}
	\end{figure}
	
	\begin{claim}[$\mathsf{A} \approx \mathsf{H}_1$] \label{clm:s-ow-a-hybrid-1}
		Recall that $\epsilon'$ is the error of $\mathsf{SUExt}$. Then
		\begin{equation}
		\#\{x \in \{0, 1\}^n : \mathsf{A}(\mathcal{P}, v_0, x, T) \not \sim_{\epsilon' r \cdot 2^{S - 1}} \mathsf{H}_1(\mathcal{P}, v_0, x, T)\} \leq 2^{n/S^c}.
		\end{equation}
	\end{claim}
	
	\begin{proof}[Proof sketch]
		The proof follows exactly the same reasoning as the proof of \cref{clm:a-hybrid-1}. The number of bad $x$ values is bounded by
		\begin{align}
		\# \text{ bad } x &\leq B \cdot 2^S \cdot 2^{n - |I'_b|} \cdot 2^{k + S + 1} \\
		&\leq 2^{3h + \sqrt{n} + O(S)} \\
		&\leq 2^{n/S^{c + 1} + \sqrt{n} + O(S)} \\
		&\leq 2^{3n/S^{c + 1}} \\
		&\leq 2^{n/S^c}
		\end{align}
		for sufficiently large $n$.
	\end{proof}
	
	\begin{figure}
		\begin{framed}
			\begin{enumerate}
				\item Initialize $v = v_0$. Repeat $r$ times, where $r \stackrel{\text{def}}{=} \lceil T/h \rceil$:
				\begin{enumerate}
					\item Let $b \in [B]$ be such that $i(v) \in I_b$. Let $I = I'_b$.
					\item Pick $y'' \in \{0, 1\}^T$ uniformly at random.
					\item Let $v = \mathcal{P} \vert_{[n] \setminus I}(v; x, y'')$.
				\end{enumerate}
				\item Output $v$.
			\end{enumerate}
			\vspace{-5mm}
		\end{framed}
		\vspace{-5mm}
		\caption{The algorithm $\mathsf{H}_2$ defining the second hybrid distribution used to prove \cref{eqn:simulate-s-ow-branching-program}. The only difference between $\mathsf{H}_1$ and $\mathsf{H}_2$ is that $\mathsf{H}_2$ feeds true randomness to $\mathcal{P} \vert_{[n] \setminus I}$, instead of feeding it a pseudorandom string from Nisan's generator.} \label{fig:s-ow-hybrid-2}
	\end{figure}
	
	The second hybrid distribution is defined by the algorithm given in \cref{fig:s-ow-hybrid-2}.
	\begin{claim}[$\mathsf{H}_1 \approx \mathsf{H}_2$] \label{clm:s-ow-hybrid-1-hybrid-2}
		For every $x$,
		\begin{equation}
		\mathsf{H}_1(\mathcal{P}, v_0, x, T) \sim_{\epsilon r} \mathsf{H}_2(\mathcal{P}, v_0, x, T),
		\end{equation}
		where $\epsilon$ is the error of $\mathsf{NisGen}$.
	\end{claim}
	
	\begin{proof}[Proof sketch]
		The proof is the same as that of \cref{clm:hybrid-1-hybrid-2}.
	\end{proof}
	
	All that remains is the final step of the hybrid argument. In this case, $\mathsf{H}_2$ actually simulates $\mathcal{P}$ with \emph{no} error. This argument is where we finally use the fact that $\mathcal{P}$ only has sequential access to its input.
	
	\begin{claim}[$\mathsf{H}_2 \sim \mathcal{P}$] \label{clm:s-ow-hybrid-2-p}
		For every $x$,
		\begin{equation}
		\mathsf{H}_2(\mathcal{P}, v_0, x, T) \sim \mathcal{P}(v_0; x, U_T).
		\end{equation}
	\end{claim}
	
	\begin{proof}[Proof sketch]
		The set $I'_b$ chosen by $\mathsf{H}_2$ excludes every index in $[n]$ that is within $h$ of $i(v)$. Therefore, each iteration of the loop in $\mathsf{H}_2$ simulates at least $h$ steps of $\mathcal{P}$. Since $r \geq T/h$, overall, $\mathsf{H}_2$ simulates at least $T$ steps of $\mathcal{P}$. But $T \geq \length(\mathcal{P})$, so we are done, just like in the proof of \cref{clm:hybrid-3-p}.
	\end{proof}
	
	\begin{proof}[Proof of \cref{thm:simulate-s-ow-branching-program}]
		By \cref{clm:s-ow-a-hybrid-1,clm:s-ow-hybrid-1-hybrid-2,clm:s-ow-hybrid-2-p} and the triangle inequality,
		\begin{equation}
		\#\{x \in \{0, 1\}^n : \mathsf{A}(\mathcal{P}, v_0, x, t) \not \sim_{\delta} \mathcal{P}(v_0; x, U_T)\} \leq 2^{n/S^c},
		\end{equation}
		where $\delta = \epsilon r + \epsilon' r \cdot 2^{S - 1}$. By our choice of $\epsilon$, the first term is at most $e^{-cS} / 2$. By our choice of $\epsilon'$, the second term is also at most $e^{-cS}/2$. Therefore, $\delta \leq e^{-cS}$.
	\end{proof}

	\section{Proof of \cref{thm:ra00-improved-generalized}: Disambiguation with advice} \label{apx:ra00-improved-generalized}
	
	We begin with randomness-efficient amplification of \cref{thm:vmp-generator}; \cref{lem:zp.uspace} is analogous to \cref{lem:zp.space}, and its proof follows the same reasoning. The details are included only for completeness.

	\begin{lemma} \label{lem:zp.uspace}
		Fix $S: \N \to \N$ with $S(n) \geq \log n$ and $\delta: \N \to [0, 1]$, both constructible in space $O(S)$. For every $L \in \mathbf{NSPACE}(S)$, there is a single-valued unambiguous algorithm $\mathcal{A}$ so that for every $x \in \{0, 1\}^n$,
		\begin{align}
		\Pr_{y \in \{0, 1\}^{\infty}}[\mathcal{A}(x, y) \in \{L(x), \botr\}] &= 1, \label{eqn:zero-error} \\
		\Pr_{y \in \{0, 1\}^{\infty}}[\mathcal{A}(x, y) = \botr] &\leq \delta(n).
		\end{align}
		Furthermore, $\mathcal{A}$ only reads the first $\log_2(1/\delta(n)) + O(S^2)$ bits of $y$ and runs in space $O(S)$.
	\end{lemma}
	
	\begin{proof}
		Let $\mathcal{R}$ be an $O(S)$-space reduction from $L$ to $\stconn$. For $x \in \{0, 1\}^n$, $\mathcal{R}(x) \in \{0, 1\}^{\bar{n}}$, where $\bar{n} = 2^{O(S)}$, and without loss of generality, $\bar{n}$ depends only on $n$. Let $s$ be the number of random bits used by $\mathsf{vMPSeededAlg}$ on inputs of length $\bar{n}$, so that $s \leq O(\log^2 \bar{n}) = O(S^2)$.
		
		Let $\ell = \lceil \log_2(1/\delta) \rceil + 2s + 2$, and let $\mathsf{GUVExt}: \{0, 1\}^{\ell} \times \{0, 1\}^d \to \{0, 1\}^s$ be the $(2s, 0.1)$-extractor of \cref{thm:guv09}, so that $d \leq O(\log \log(1/\delta) + \log S)$. On input $x \in \{0, 1\}^n, y \in \{0, 1\}^\ell$:
		\begin{enumerate}
			\item For every $z \in \{0, 1\}^d$:
			\begin{enumerate}
				\item Let $a = \mathsf{vMPSeededAlg}(\mathcal{R}(x), \mathsf{GUVExt}(y, z))$.
				\item If $a \neq \botr$, halt and output $a$.
			\end{enumerate}
			\item Halt and output $\botr$.
		\end{enumerate}
		
		Clearly, this algorithm runs in space $O(S + d)$. Since $\delta$ is constructible in space $O(S)$, its denominator must have at most $2^{O(S)}$ digits. Therefore, $\delta \geq 2^{-2^{O(S)}}$ and $d \leq O(S)$, so the algorithm runs in space $O(S)$. Furthermore, it is clearly single-valued unambiguous, and it is ``zero-error'', i.e., \cref{eqn:zero-error} holds. Finally, by \cref{prop:zuc97}, the number of $y$ such that $\mathsf{vMPSeededAlg}(\mathcal{R}(x), \mathsf{GUVExt}(y, z)) = \botr$ for every $z$ is at most $2^{2s + 2}$, and hence the probability that the algorithm outputs $\botr$ is at most $\frac{2^{2s + 2}}{2^{\ell}} \leq \delta$.
	\end{proof}
	
	\begin{proof}[Proof of \cref{thm:ra00-improved-generalized}]
		Let $\mathcal{A}$ be the algorithm of \cref{lem:zp.uspace} with $\delta < 2^{-n}/\epsilon$. Let $m = m(n)$ be the number of random bits used by $\mathcal{A}$. Let $\mathcal{B}$ be the algorithm witnessing the fact that $L$ is within $\epsilon$ of $\mathbf{USPACE}(S)$.
		
		Given input $x \in \{0, 1\}^n$ and advice $a \in \{0, 1\}^m$, compute $a = \mathcal{A}(x, a)$. If $a \neq \botr$, output $a$. If $a = \botr$, output $\mathcal{B}(x)$. This algorithm clearly runs in $O(S)$ space, uses $n - \log_2(1/\epsilon(n)) + O(S^2)$ bits of advice, and is unambiguous (in fact, single-valued unambiguous).
		
		Now we argue that there is some advice string such that the algorithm succeeds on all inputs. Let $S \subseteq \{0, 1\}^n$ be the set of inputs on which $\mathcal{B}$ fails. Consider picking an advice string $a$ uniformly at random. For each string $x \in S$, $\Pr_a[\mathcal{A}(x, a) = \botr] \leq \delta$. Therefore, by the union bound, the probability that there is some $x \in S$ such that $\mathcal{A}(x, a) = \botr$ is at most $|S| \delta = \epsilon \cdot 2^n \cdot \delta < 1$. Therefore, there is \emph{some} advice string such that the algorithm succeeds on all inputs in $S$. Finally, for \emph{any} advice string, the algorithm succeeds on all inputs in $\{0, 1\}^n \setminus S$ by \cref{eqn:zero-error}.
	\end{proof}
\end{document}